\documentclass{article}

\usepackage{amssymb}										\usepackage{amsmath}								
\usepackage{amsthm}											\usepackage{braket}											\usepackage{amsfonts}										\usepackage{calc}											\usepackage{amssymb,stmaryrd,gensymb}										\usepackage{bbm}											\usepackage{mathrsfs}										\usepackage{graphicx}										\usepackage[usenames,dvipsnames,svgnames,table]{xcolor} 	\usepackage[breaklinks]{hyperref}						\usepackage[normalem]{ulem}										\usepackage[T1]{fontenc}									\usepackage[utf8]{inputenc}								\usepackage{pict2e}
\usepackage{tikz}
\usepackage{caption,cite}
\usepackage{ytableau}

					\newcommand{\sTD}[1]{\overline{\underline{#1}}}
\newcommand{\sT}[1]{\overline{#1}}
\newcommand{\sD}[1]{\underline{#1}}
\newcommand{\sTDL}{\,\overbar{\sD{\Lambda}}}
\newcommand{\sTL}{{\,\vphantom{\sTDL}\overbar{\Lambda}}}
\newcommand{\sDL}{{\,\vphantom{\sTDL}\sD{\Lambda}}}
\newcommand{\sLs}{ { {\,\vphantom{\sTDL}\Lambda} }^{0} }
\newcommand{\sL}{{\vphantom{\sTDL}\Lambda}}
\newcommand{\sTd}{\sT{d}}
\newcommand{\sDd}{\sD{d}}
\newcommand{\omegahat}{\hat{\omega}}

\newcommand{\sTrig}{
	\begin{tikzpicture}[scale=.145, remember picture]	\begin{scope}[overlay]
		\draw[semithick] (0,.4) circle [radius=1];
		\clip (0,.4) circle [radius=1];
		\draw[thick] (0,-1.45) -- (0,1.45);
	\end{scope}
	\end{tikzpicture}
}
\newcommand{\sCercle}{
	\begin{tikzpicture}[scale=.145, remember picture]	\begin{scope}[overlay]
		\draw[semithick] (0,.4) circle [radius=1];
		\clip (0,.4) circle [radius=1];
		\draw[thick] (-1,.4) -- (1,.4);
	\end{scope}
	\end{tikzpicture}
}

\newcommand{\sDouble}{
	\begin{tikzpicture}[scale=.145, remember picture]	\begin{scope}[overlay]
		\draw[semithick] (0,.4) circle [radius=1];
		\clip (0,.4) circle [radius=1];
		\draw[thick] (0,-1.45) -- (0,1.45);
		\draw[thick] (-1,.4) -- (1,.4);
	\end{scope}
	\end{tikzpicture}
}
\newcommand{\sFill}[1]{
		\begin{tikzpicture}[scale=.145, remember picture]	\begin{scope}[overlay]
		\draw[clip] (0,.4) circle [radius=1];
	\end{scope}
	\begin{scope}[overlay]
		\draw[clip] (0,.4) circle [radius=1];
		\node[scale=.6] at (0,.4) {${#1}$};
	\end{scope}
	\end{tikzpicture}
}
\newcommand{\yT}{\none[\sTrig]}
\newcommand{\yC}{\none[\sCercle]}
\newcommand{\yTC}{\none[\sDouble]}
\newcommand{\yF}[1]{\none[{\sFill{#1}}]}
\newcommand{\superYsmall}[1]{\, \scalebox{.5}{\begin{ytableau}#1 \end{ytableau} }}
\newcommand{\superY}[1]{\begin{array}{c} \scalebox{1}{\begin{ytableau}#1 \end{ytableau} } \end{array}}
\newcommand{\superYbig}[1]{\begin{array}{c} \scalebox{2}{\begin{ytableau}#1 \end{ytableau} } \end{array}}
\newcommand{\SPar}[1]{\text{SPar}(#1)}

\newcommand{\N}[1]{$\mathcal{N}=#1$}
\newcommand{\overbar}[1]{\mkern 1.5mu\overline{\mkern-1.5mu#1\mkern-1.5mu}\mkern 1.5mu}

\newcommand{\anote}[2][Sans Nom]{\begin{note}[#1] \co{#2} \end{note}}

\newtheorem{theorem}{Theorem}[section]
\newtheorem{corollary}{Corollary}[section]
\newtheorem{proposition}{Proposition}[section]
\newtheorem{lemma}[theorem]{Lemma}

\theoremstyle{definition}
\newtheorem{example}{Example}[section]
\newtheorem{definition}{Definition}[section]
\theoremstyle{remark}
\newtheorem{remark}{Remark}[section]
\newtheorem{note}{Note}[section]

\def\co#1{{\color{red}#1}}
\def\cy#1{{\color{cyan}#1}}

\def\LAV#1{{\color{violet}#1}}

\let\ll\llbracket

\let\rr\rrbracket

\let\la\lambda
\let\La\Lambda
\let\a\alpha
\let\ti\tilde

\def\O{{\mathcal O}}

\def\Nm{{\mathcal N}}
\def\D{{\mathcal D}}
\let\u\underline
\let\o\overline
\let\w\wedge
\let\nb\nonumber

\let\d\partial
\def\sTt{\sT{\tau}}
\def\sDt{\sD{\tau}}
\let\Om\Omega
\let\om\omega
\let\ta\theta
\def\beq{\begin{equation}}
\def\eeq{\end{equation}}

\let\ep\epsilon
\let\l\left
\let\r\right

	\ytableausetup{smalltableaux}
	\title{$\mathcal{N}\geq  2$ symmetric superpolynomials}

\author{ L. Alarie-V\'ezina\thanks{D\'epartement de physique, de g\'enie physique et
d'optique, Universit\'e Laval,  Qu\'ebec, Canada,  G1V 0A6; ludovic.alarie-vezina.1@ulaval.ca},
 L. Lapointe\thanks{Instituto de Matem\'atica y F\'{\i}sica, Universidad de
Talca, 2 norte 685, Talca, Chile; lapointe@inst-mat.utalca.cl }
and  P. Mathieu \thanks{D\'epartement de physique, de g\'enie physique et
d'optique, Universit\'e Laval,  Qu\'ebec, Canada,  G1V 0A6; pmathieu@phy.ulaval.ca}
			}

\begin{document}	

	\maketitle
\begin{abstract}
The theory of symmetric functions has been extended to the case where each variable is paired with an anticommuting one. The resulting expressions, dubbed superpolynomials, provide the natural  ${\mathcal N}=1$ supersymmetric version of the classical bases of symmetric functions.  Here we consider the case where more than one independent anticommuting variable is attached to each ordinary variable.  First, the ${\mathcal N}=2$ super-version of the monomial, elementary, homogeneous symmetric functions, as well as the power sums, are constructed systematically (using an exterior-differential formalism for the multiplicative bases), these functions being now indexed by a novel type of superpartitions.
Moreover, the scalar product of power sums turns out to have a natural ${\mathcal N}=2$ generalization which preserves the duality between the monomial and homogeneous bases. All these results are then generalized to an arbitrary value of ${\mathcal N}$. Finally, for ${\mathcal N}=2$, the scalar product and the homogeneous functions are shown to have a one-parameter deformation, a result that prepares the ground for the yet-to-be-defined ${\mathcal N}=2$ Jack superpolynomials.
\end{abstract}

		\graphicspath{./figures}
	
\section{Introduction}
\subsection{The theory of ${\mathcal N}=1$ superpolynomials}
The classical theory of symmetric polynomials
turns out to have a rich and nontrivial extension when additional Grassmannian variables are involved. More precisely, in addition to be functions of  the usual variables $x_1,\cdots ,x_N$, these generalized polynomials (thereby dubbed superpolynomials) also depend upon anticommuting variables $\theta_1,\cdots,\theta_N$, where $\theta_i\theta_j=-\theta_j\theta_i$, (so that $\theta_i^2=0$). The ``symmetric'' requirement amounts to enforce the symmetry under the simultaneous interchange of pairs $(x_i,\theta_i)$ and $(x_j,\theta_j)$. Two simple examples of symmetric superpolynomials for $N=2$ are
\begin{equation} \label{exa}\theta_1x_1^2x_2+\theta_2x_2^2x_1\qquad  {\rm and} \qquad \theta_1\theta_2(x_1^2x_2-x_1x_2^2).\end{equation}
(The symbol indicating the number of variables of each type, $N$, should not be confused with that giving the number of supersymmetries, $\mathcal N$.)
Note that symmetric superpolynomials that do not depend upon anticommuting variables are ordinary symmetric polynomials (equivalently, these are ${\mathcal N}=0$ superpolynomials.\\

The classical bases of symmetric polynomials, namely, the monomial symmetric functions, the power sums, as well as the elementary and homogeneous symmetric functions,  all have a natural extension to the super-case \cite{classicalN1}. Even more remarkably, the Schur, Jack and Macdonald polynomials can also be generalized to superspace.
An amazing feature of this extension is that most of the properties of these symmetric polynomials  are  preserved or extended in a natural way, in spite of the complications
 brought in by the presence of anticommuting variables.\\
 
Let us substantiate the last statement (thereby fixing the targets for the program initiated here).  Like their ${\mathcal N}=0$ counterpart, the Jack superpolynomials are defined by two conditions: triangularity (in the monomial basis) and orthogonality. The latter can be defined either from
 the ``combinatorial'' scalar product (defined in terms of the power sums) \cite{DLMadv} or the ``physical'' one (defined by 
a Selberg-type integral) \cite{DLMsJack}.  The qualifier ``physical'' refers to 
the framework of an alternative characterization of the Jack superpolynomials as the eigenfunctions of a supersymmetric quantum-mechanical $N$-body problem. Among other noteworthy features of the Jack polynomials that are lifted to the super-level, we single out the remarkable clustering property  \cite{FJMM} at negative fractional values of the free parameter $\alpha$ (for the so-called admissible partitions) \cite{DLM0,DG}  and the connection with singular vectors of the (super)Virasoro algebra  \cite{MY,DLMsvir,ADM}.\\

Most of these results 
can be generalized to the Macdonald superpolynomials \cite{BF1,BF2}. These are specified by two parameters $q$ and $t$ (where the Jack limit is $q=t^\a,\ t\to 1$). 
 Setting $q=t$ yields a one-parameter family of Schur-like superpolynomials.
 It happens that these display remarkable
  combinatorial properties for two special cases: $t=0$ and $t\to \infty$. 
The Pieri rules for these two versions of the Schur superpolynomials have been obtained, revealing a  rich underlying combinatorial structure \cite{BM,JL}.
 Now, a truly remarkable observation is that the (plethystic deformation of the) Macdonald superpolynomials appear to have a positive decomposition in the Schur basis (the $t=0$ version), 
 which is recognized as  a supersymmetric extension of the original Macdonald conjecture \cite{BF1} (which holds when the degree in the Grassmannian variables is large enough \cite{BLM}).

\subsection{The ${\mathcal N}=2$ superpolynomials: a first encounter}
In the present work, we open  a new section to this program by further extending the classical symmetric polynomials to be  functions of two independent sets of anticommuting variables $\theta_1,\cdots \theta_N$ and $\phi_1,\cdots,\phi_N$. Here, we not only impose that the variables in each set anticommute among themselves:
\begin{equation} \theta_i\theta_j=-\theta_j\theta_i\,,\qquad \phi_i\phi_j=-\phi_j\phi_i,\end{equation} but we also require the anticommutativity of two of the distinct types    of variables:
\begin{equation}\label{tap}\theta_i\phi_j=-\phi_j\theta_i.\end{equation}
(Note that the commuting variables commute with the anticommuting ones: $x_i\ta_j=\ta_j x_i$ and $x_i\phi_j=\phi_jx_i$.)
Moreover, we generalize the symmetry requirement to be the invariance under the interchange of two triplets $(x_i,\phi_i,\theta_i)\leftrightarrow (x_j,\phi_j,\theta_j)$.  We call the resulting objects  ${\mathcal N}=2$ symmetric superpolynomials (and refer to those superpolynomials with no $\phi_i$ variables as ${\mathcal N}=1$ superpolynomials).\cite{Note}\\

Let us return for a moment to the ${\mathcal N}=1$ symmetric superpolynomials. A simple basis in which they can be expanded is the one generated by the product of power sums. Obviously, to the usual power sum $p_k$ we need to adjoin a fermionic companion  denoted $\tilde p_k$, where
\begin{equation} p_k=\sum_{i=1}^N x_i^k\qquad\text{and}\qquad \tilde p_k=\sum_{i=1}^N\theta_i x_i^k.\end{equation}
The two examples displayed in eq. \eqref{exa} can be rewritten as $\tilde p_2p_1-\tilde p_3$ and $\tilde p_2\tilde p_1$ respectively.\\

How does this generalize to the ${\mathcal N}=2$ case?  One certainly needs to add the $\phi$-version of $\tilde p_k$:
\begin{equation} \tilde p'_k=\sum_{i=1}^N \phi_ix_i^k.\end{equation}
(Note that a different notation will be used in the main text for the power sums.)
Having added two sets of anticommuting variables naturally suggests that adding two independent fermionic power sums would be sufficient. However, consider the following simple example:
\begin{equation} \label{exb}\phi_1\theta_1 x_1+\phi_2\theta_2 x_2.\end{equation}
This certainly satisfies our symmetry requirement. But it is straightforward to see that it cannot be expressed  in terms of $\tilde p_1$ and $ \tilde p'_1$. Therefore, a fourth type of power sums needs to be introduced, namely:
\begin{equation} \hat p_k=\sum_{i=1}^N \phi_i\theta_ix_i^k.\end{equation}												Observe that the  example \eqref{exb} is simply $\hat p_1$ for $N=2$. \\

Four varieties of power sums, two  fermionic and two bosonic, are thus necessary for the formulation of this multiplicative basis. This is also the case for the other multiplicative bases (elementary and homogeneous).\\
			
	Symmetric ${\mathcal N}=2$ superpolynomials are
		labelled by ${\mathcal N}=2$ superpartitions.			The occurrence of four types of power sums suggests  that the superpartitions are composed of four partitions. The splitting of these four types into two bosonic and two fermionic ones further entails that two of these partitions -- the one associated with the product of the $\tilde p_k$ and the other with that of the  $\tilde p'_k$ -- will have a restriction, namely that their parts should be distinct.\\

The primary objective of this work is to present a detailed analysis of the ${\mathcal N}=2$ supersymmetric version of the classical bases $m_\lambda, e_\lambda, h_\lambda$ and $p_\lambda$ (using the standard notation in which $\lambda$ is a partition, cf \cite{Macdonald1998}), including their generating functions and the relationships between them.  Moreover, we determine the proper modification of the scalar product in the power-sum basis, with respect to which the ${\mathcal N}=2$ generalized versions of $m_\lambda$ and $h_\lambda$ are dual bases. \\

Somewhat surprisingly, the $\Nm=2$ classical bases can be generalized straightforwardly to a generic value of $\Nm$.  Here again, we can fix the scalar product and still preserve the duality between monomial and homogeneous polynomials.\\


Back to the $\Nm=2$ context. The
scalar product is found to have a natural deformation by a free parameter, denoted $\a$. We also observe that the $\alpha$-deformation of the homogeneous functions can be defined in  a natural way, by still preserving their duality with the monomials. 

\subsection{Motivation}
The extension of the classical theory of symmetric functions to superpolynomials with an arbitrary number of anticommuting partner variables is of interest on its own.
However, the ultimate objective of this work is to pave the way for the construction of the $\Nm=2$ Jack superpolynomials.
We expect those to be defined
by directly extending the  $\Nm=0,1$ definition to the $\Nm=2$ case, namely, in terms of two conditions: triangularity in the monomial basis and orthogonality.
The scalar product with respect to which we expect the yet-to-be-defined ${\mathcal N}=2$ Jack superpolynomials to be orthogonal is precisely the $\alpha$-deformation of the power sums scalar product just alluded to.
Moreover, the three companions of the deformed $h_n$ are candidates for the particular ${\mathcal N}=2$ Jack superpolynomials that are indexed by a  superpartition represented by a single row diagram.
However, the generic construction of the ${\mathcal N}=2$ Jack superpolynomials is still a work in progress due to some difficulties that are discussed at the end of this article.

\subsection{Organization of the article}
	The superpartitions are introduced in section 2.1, together with their diagrammatic representation.  The monomial super functions, introduced in section 2.4, are readily obtained from the structure of the superpartition and the symmetrization process.\\

In section 3.1 we introduce our main technical tool for the demonstration of many results pertaining to the multiplicative bases, a trick  lifted from
\cite{classicalN1} but barely exploited there. The idea is to interpret a product of $p$ Grassmannian variables $\theta_i$
times a polynomial in the $x_j$, as a $p$-form and to introduce a differential operator $\underline{d}$ that transforms $p$-forms into $p+1$-forms using the formal identification $\underline{d}x_i\longleftrightarrow \theta_i$. A similar operator, denoted as $\overline{d}$, acts on polynomials involving  
products of $\phi_i$'s, this time with $\overline{d}x_i\longleftrightarrow \phi_i$. \\

Let $f_k$ stand for either $p_k, e_k$ or $h_k$. Then the three companions
 of $f_k$ are generated from it by acting with $\underline{d}, \overline{d}$ and $\overline{d}\, \underline{d}$.  Moreover, the $f_k$ generating function, when acted upon by a deformation of the product $\overline{d}\, \underline{d}$, denoted ${\mathcal O}$ (namely, this is the operator $(1+\overline{{\mathcal D}}) (1+\underline{{\mathcal D}})$ defined in \eqref{def:bigdiffO}), produces the generating function for the ${\mathcal N}=2$ $f$-type multiplicative basis.  This is the subject of sections 3.2-3.3. Next,     simple recursion relations for the four components of the three multiplicative bases are derived from the relationships between their generating functions  in \S 3.4. Finally, in \S 3.5, the ${\mathcal N}=2$ reproducing kernel is written down, which yields the super-version of the power sums scalar product.  The validity of the presented kernel is then confirmed by the verification of the dual nature of the monomial and homogeneous  ${\mathcal N}=2$ versions.\\

In section 4, we present the extension of the previous results for an arbitrary number of supersymmetries.\\

The concluding section contains some preliminary remarks in relation with the still-undefined $\Nm=2$ Jack superpolynomials. In particular, we 
introduce the ${\mathcal N}=2$ version of the 
$g_\lambda^{(\alpha)}$ symmetric function (with $g_\lambda^{(1)}=h_\lambda$). 
These generalized   $g_\lambda^{(\alpha)}$'s are shown to be orthogonal to the super-monomial functions with respect to the one-parameter extension of the ${\mathcal N}=2$  scalar product.\\

{Three} appendices complete this article. The proof of  a simple identity is the subject of the first one. {In Appendix B, we present an algorithm describing the decomposition of the product of two monomials in the monomial basis.  This construction is then used in Appendix C to demonstrate  the completeness of the three multiplicative bases, ensuring thus their ``basis'' qualifier.}

 	\section{Superpartitions and monomial superfunctions}

	\subsection{Superpartitions.} 
				We recall that an $\mathcal{N}=1$ superpartition $\Lambda$ 		is composed of two partitions, $\Lambda = (\Lambda^a; \Lambda^s)$, where $\Lambda^a$ has distinct parts and
 $\Lambda^s$ is an ordinary partition \cite{classicalN1}. The superpartition $\Lambda$ is said to belong to the $(n|m)$ sector if the number of parts of $\Lambda^a$, called the \emph{fermionic degree}, is $m$ and the degrees of $\Lambda^a$ and $\Lambda^s$ add up to $n$.  In other words, a superpartition belonging to the  $(n|m)$ sector is such that 
		\begin{align}
		&	\Lambda =(\La^a;\La^s)= (\Lambda_1, \cdots, \Lambda_m;\; \Lambda_{m+1}, \cdots, \Lambda_N),\quad {\rm with} \quad \sum_{i=1}^N \Lambda_i = |\Lambda^a|+|\Lambda^s|= n,
		\nonumber\\
		&	\text{and where~~}\; \Lambda_1 > \Lambda_2 > \cdots >\Lambda_m \geq 0  \quad {\rm and} \quad
			 \Lambda_{m+1} \geq \Lambda_{m+2} \geq \cdots \geq \Lambda_N.		\end{align}
	The upper index $a$ and $s$ denote respectively antisymmetric and symmetric components.  Note that these adjectives
refer to the structure of the superpolynomials they label and not to the partitions \emph{per se}.
For the same reason, the parts $\Lambda_1, \cdots, \Lambda_m$ are often called the  \emph{fermionic}  parts.	\\	
		In the $\mathcal{N}=2$ case, as mentioned in the Introduction, we need two extra partitions to keep track of the
interplay between the two different types of fermionic variables and the commuting variables.  Hence, the following definition.

				\begin{definition}
			An $\mathcal{N}=2$ superpartition $\Lambda$ 			 is an ordered sequence of nonnegative integers, in which the integers can also be marked by an overline, an underline, or an overline and an underline at the same time.
			 			  Parts that are only overlined (resp. only underlined) need to be distinct and can be equal to zero, while parts that are either unmarked or both  overlined and underlined (refered to as \emph{bilined} in the remainder of the article) have no restriction and can also be equal to zero.
			  			   A generic element of the superpartition $\Lambda$ is denoted $\Lambda_i$ whether it is marked or not and its numerical value is denoted by $|\Lambda_i|$. The parts appear in decreasing order, and for parts of equal value, the order (denoted $\succ$) is the following: bilined, overlined, underlined, and not marked, i.e.,
			   \beq \Lambda = (\Lambda_1, \Lambda_2, \dots, \Lambda_{\ell}), \qquad|\La_i|\geq |\La_{i+1}|\qquad \text{and}\qquad\sTD{a}\succ \sT{a}\succ \sD{a}\succ a\quad  \text{for}\; a=|\La_i|.\label{orderP} \eeq
			   
			We also introduce four ordered sets (which we will refer to as the constituent partitions) that collect the different types of parts. Those sets are denoted in the obvious manner $\sTD{\Lambda}, \sT{\Lambda}, \sD{\Lambda}, \Lambda^0$ so that for instance parts that are bilined are collected in the set $\sTD{\Lambda}$,
			  while unmarked parts are collected in the set $\Lambda^0$. (Note that in  \cite{classicalN1}, $\sTD{\Lambda}$ refers to the fermionic degree, a notation that we shall no longer use.)  Rephrased compactly, these conditions are: 			 \begin{align}
&	\Lambda = (\Lambda_1, \Lambda_2, \dots, \Lambda_{\ell}), \quad
			 	\text{where }\; |\La_i| \geq |\Lambda_{i+1}|, \nonumber \\
			 	&
				\text{and } 
			 	\forall\; \Lambda_i, \Lambda_j \in \sTDL\; \text{or~} \Lambda_i, \Lambda_j \in \sLs\text{ with } i>j,\, \text{we have~}|\La_i| \geq |\Lambda_j|\geq 0, \\
			 	&\phantom{\text{and }}
				\forall\; \Lambda_i, \Lambda_j \in \sTL \; \text{or~} \Lambda_i, \Lambda_j \in \sDL\;\;\text{ with } i>j,\, \text{we have~} |\La_i| >|\La_j| \geq 0.
			 			 			 \end{align}

		The total number of parts in a superpartition (not counting the zero parts of $\Lambda^0$ but keeping track of those of $\sTD{\La}$) defines its length, written $\ell(\Lambda)$. Likewise we denote by $\ell(\sTDL)$, $\ell(\sTL)$, $\ell(\sDL)$ and $\ell(\sLs)$ the number of parts in each of the constituent partitions  (not counting again the zero parts of $\Lambda^0$). Let $\Lambda$ be a generic superpartition, we define
		\begin{align}
			\sT{m}_\Lambda 		&:= \ell(\sT{\Lambda}) + \ell(\sTD{\Lambda}), \qquad
			\sD{m}_\Lambda 		:= \ell(\sD{\Lambda}) + \ell(\sTD{\Lambda}), \qquad
			|\Lambda| 			:= \sum_{i=1}^{\ell(\Lambda)}|\La_i|. 
		\end{align}
		These three integers allow us to define the so-called sectors. Given a superpartition $\Lambda$ with 
		\begin{align}
			\sT{m}_\Lambda 		&= \sT{m},\quad
			\sD{m}_\Lambda 		= \sD{m}, \quad
			|\Lambda| 			= n, 
		\end{align}
		we say that $\Lambda$ is a superpartition in the sector $(n|\sT{m}, \sD{m})$ and write $\Lambda \vdash (n|\sT{m}, \sD{m})$. The set composed of all superpartitions in the sector $(n| \sT{m}, \sD{m})$ is denoted $\text{SPar}(n|\sT{m}, \sD{m})$, and the set of all \N{2} superpartitions is simply denoted $\text{SPar}$. 
 		\end{definition}

		\begin{remark}
		When either $\sT{m}$ or $\sD{m}$ is 0 we recover $\mathcal{N}=1$ superpartitions and when both are zero we recover standard partitions, that is 
		\begin{align}
			\text{SPar}(n|\sT{m},0)=\text{SPar}(n|\sT{m}),\qquad
			\text{SPar}(n|0,\sD{m})=\text{SPar}(n|\sD{m}), \qquad
			\text{SPar}(n|0,0)=\text{Par}(n). 
		 \end{align}
		\end{remark}
		
		\begin{example} \label{ex:superpartitions}
			Let $\Lambda$ be a superpartition,
				\begin{align}
					\Lambda =(\sTD{4}, \sT{4}, \sD{3}, 3, \sTD{2},  \sT{2},\sD{2}, 1, \sTD{0}, \sTD{0}, \sT{0}). 				\end{align}
			Then, from the previous definitions, we have 
			\begin{gather}
				|\Lambda| 		= 4 + 4 + 3 +3 +2 +2 +2 +1 +0 +0 +0= 21\nonumber \\
				\ell(\Lambda) 	= 11,\; \ell(\sTDL) = 4,\; \ell(\sTL)=3,\; \ell(\sDL)= 2,\;
				\sT{m}			= 7,\; 
				\sD{m}			= 6,\nonumber\\
						\Lambda \in \SPar{21|7,6}.
			\end{gather}
		\end{example}

\subsection{Diagrammatic representation}
		The superpartition notation involving marks suggests the following decorated Ferrers diagram representation. As usual, every part is represented as a row of boxes corresponding to its value but if the part is marked, we add the proper decoration at the end of the row. Overlined parts are decorated with a circle containing a vertical line (v-circle), underlined parts are decorated with a circle containing a horizontal line (h-circle) and bilined parts are decorated with a circle containing both a horizontal and a vertical line  (vh-circle). We add the convention that when given the choice, the v-circle should appear on top of the h-circle. This is probably much clearer with the following example.\\

		Taking again the  superpartition  used in Example \ref{ex:superpartitions}, we have :
		\begin{align}
						\Lambda =(\sTD{4}, \sT{4}, \sD{3}, 3, \sTD{2},  \sT{2},\sD{2}, 1, \sTD{0}, \sTD{0}, \sT{0}) \quad \longleftrightarrow \quad 
						\scalebox{1.5}{$\superY{
												\,&\,&\,&\,&\yTC\\
												\,&\,&\,&\,&\yT\\
												\,&\,&\,&\yC\\
												\,&\,&\,\\
												\,&\,&\yTC\\
												\,&\,&\yT\\											
												\,&\,&\yC\\
												\,\\
												\yTC\\
												\yTC\\
												\yT
												}$}
		\end{align}
		This example seems a bit ill-selected in order to promote the diagrammatic representation, but it does exhibit all the special features of the $\mathcal{N}=2$ superpartitions. The elegance of the notation is somewhat more manifest in the following example.

		\begin{example} \label{ex:spar2-1-1}
			The set of superpartitions in the sector $(2|1,1)$ is given by
			\begin{align}
				\SPar{2|1,1} = 	
				&\left\lbrace
				(\sTD{2}), (\sT{2},\sD{0}),(\sD{2},\sT{0}), (2,\sTD{0}), (2,\sT{0}, \sD{0}), (\sTD{1}, 1), \right. \\
				 &\; \left. (\sT{1},\sD{1}), (\sT{1},1,\sD{0}), (\sD{1}, 1, \sT{0}), (1,1,\sTD{0}), (1,1,\sT{0}, \sD{0})
				\right\rbrace.
			\end{align}
			Their diagrammatic representation, in the same order, is
			\begin{gather*}
				\SPar{2|1,1} \longleftrightarrow \left\lbrace 
					\superY{
								\,&\,&\yTC
					},
					\superY{
								\,&\,&\yT\\
								\yC
					},
					\superY{
								\,&\,&\yC\\
								\yT
					},
					\superY{
								\,&\,\\
								\yTC
					},
					\superY{
								\,&\,\\
								\yT\\
								\yC
					},
					\superY{
								\,&\yTC\\
								\,
					},\right. \\
					\left.
					\superY{
								\,&\yT\\
								\,&\yC
					},
					\superY{
								\,&\yT\\
								\,\\
								\yC
					},
					\superY{
								\,&\yC\\
								\,\\
								\yT
					},
					\superY{
								\,\\
								\,\\
								\yTC
					},
					\superY{
								\,\\
								\,\\
								\yT\\
								\yC
					}
					\right\rbrace
			\end{gather*}
		\end{example}

		The restrictions on the constituent partitions of the superpartition lead to the following generating function, where $p(n,  \sT{m},\sD{m})$ is the number of superpartitions in the sector $(n|\sT{m},\sD{m})$,

		\begin{align}
			\sum_{n, \sT{m},\sD{m}} p(n,  \sT{m},\sD{m})q^n \xi^{\sT{m}}\gamma^{\sD{m}}  &= \prod_{n=0}^\infty \frac{(1+\xi q^n)(1+\gamma q^n)}{(1-\xi \gamma q^n)(1-q^{n+1})}\\ 
			&=(1+\gamma+\xi+2 \xi \gamma+\xi {\gamma}^{2}+{\xi}^{2}\gamma+2 {\xi}^{2}{\gamma}^{2}+\cdots)q^0 \nonumber \\
			&\quad + (1+2 \gamma+2 \xi+{\gamma}^{2}+5 \xi \gamma+{\xi}^{2}+4 \xi {\gamma}^{2}+\cdots)q^1 \nonumber \\
			&\quad + (2+4 \gamma+4 \xi+2 {\gamma}^{2}+11 \xi \gamma+2 {\xi}^{2}+10 \xi {\gamma}^{2} + \cdots)q^2 +\cdots		\end{align}
		For instance, the coefficient of $\xi \gamma q^2$  in the last expression,  namely $11$, is the cardinality of $\SPar{2|1,1}$. The generating function highlights the fact that there is an infinite number of superpartitions of fixed degree $n$ which is a feature that does not appear in the \N{1} superpartition case. This very point is linked to the fact that $\sTDL$ can contain an arbitrary number of $0$'s.\\

		 \begin{remark}
		 	It was noted in \cite{classicalN1} that \N{1} superpartitions are equivalent to standard MacMahon diagrams (see for instance \cite{Pak2006}, section 2.1.3), which are also equivalent to overpartitions, as introduced in \cite{Corteel2004}. The only difference between \N{1} superpartitions and overpartitions is that $\bar a$ is denoted
$\overline{a+1}$ in the latter case. There is a similar relation between \N{2} superpartitions and the overpartition pairs introduced in \cite{Lovejoy2006}.  
The latter are composed of two ordinary partitions and two partitions with distinct parts. This is also the composition of  \N{2} superpartitions except that again  marked letters are shifted by one unit in overpartition pairs.
								 \end{remark}

	\subsection{Symmetric superpolynomials}

												The coordinates in \N{2} superspace take the form $(x_1, \dots, x_N; \phi_1, \dots, \phi_N; \theta_1, \dots, \theta_N)$, where Greek letters denote Grassmann numbers. 
		We now define $\mathscr{P}(\mathbb{A})$, the ring of polynomials in $x, \phi, \theta$ with coefficients in some base field $\mathbb{A}$. This construction has a tri-grading, that is
		\begin{align}
			\mathscr{P} = \bigoplus_{n, \sT{m},\sD{m}} \mathscr{P}_{(n|, \sT{m},\sD{m})},
		\end{align}
		where $\mathscr{P}_{(n|\sT{m},\sD{m})}$ is the finite dimensional module composed of all polynomials of degree $n$ in the variables $x$, 
degree $\sT{m}$ in the variables $\phi$ and degree $\sD{m}$ in the variables $\theta$.  \\

		We now define the \N{2} symmetric superpolynomials. Given the polynomial ring in superspace $\mathscr{P}$, the module of symmetric \N{2} superpolynomials in $N$ variables with coefficients in $\mathbb{A}$, denoted $\mathscr{P}^{S_N}(\mathbb{A})$, is the subset of $\mathscr{P}(\mathbb{A})$ made out of polynomials that are invariant under the \emph{simultaneous} exchange of the variables $x_i, \phi_i, \theta_i$. Put more simply, let $f$ be an arbitrary superpolynomial in $\mathscr{P}$, then $f$ is symmetric if and only if $f(x_k, x_l, \phi_k, \phi_l, \theta_k, \theta_l) = f(x_l, x_k, \phi_l, \phi_k, \theta_l, \theta_k)$ for all $k,l$, where for simplicity we only show the dependency in the variables indexed by $k$ and $l$ in the polynomial. \\ 

		More formally, let us introduce the operators $K_{ij}, \sT{\kappa}_{ij}$ and $ \sD{\kappa}_{ij}$ that generate the action of the symmetric group $S_N$ over the different variables : 
												\begin{align}
			K_{kl} : x_k \longleftrightarrow x_l, \quad
			\sT{\kappa}_{kl} : \phi_k \longleftrightarrow \phi_l,\quad
			\sD{\kappa}_{kl} : \theta_k \longleftrightarrow \theta_l.
		\end{align}
		We also define the operator $\mathcal{K}_{ij}$ that generate the diagonal action of $S_N$, that is
		\begin{align}
			\mathcal{K}_{kl} f(x_k, x_l, \theta_k, \theta_l, \phi_k, \phi_l) = K_{kl} \,\sT{\kappa}_{kl} \,\sD{\kappa}_{kl} \,f(x_k, x_l, \phi_k, \phi_l, \theta_k, \theta_l) = f(x_l, x_k, \phi_l, \phi_k, \theta_l, \theta_k). 
		\end{align}
A polynomial $f(x,\phi,\theta) \in \mathscr{P}$ then belongs to $\mathscr{P}^{S_N}$ if and only if 
		\begin{align}
			\mathcal{K}_{ij}f(x,\phi,\theta) =f(x,\phi,\theta)\quad \forall \quad i,j=1, \dots, N.
		\end{align}
 				\begin{example} Let 
			\begin{align} \label{eq:exampleSymFunction}
				f(x,\phi,\theta) = \phi_1 \phi_2 \theta_3 x_1^2 + \phi_1 \phi_3 \theta_2 x_1^2 + \phi_2 \phi_3 \theta_1 x_2^2 
				- \phi_1 \phi_2 \theta_3 x_2^2 - \phi_1 \phi_3 \theta_2 x_3^2  - \phi_2 \phi_3 \theta_1 x_3^2,
			\end{align}
			then, $f(x,\phi,\theta)$ is a \N{2} symmetric superpolynomial of degree $(2|1,2)$ since it is invariant 
under any exchange of the indices.  Since $N=3$, we have that $f(x,\phi,\theta) \in \mathscr{P}^{S_3}(\mathbb{Z})$.
		\end{example}
As in the case of symmetric functions, we generically consider that polynomials have an infinite number of variables for formal calculations, and we shall denote by $\mathscr{P}^{S_\infty}(\mathbb{A})$ the ring of symmetric functions over $\mathbb{A}$ in \N{2} superspace.

	\subsection{Monomial basis} 
	
		In view of defining the monomial basis, we introduce the following compact notation:
	\begin{align}		[\phi;\theta]_{\sL} = \phi_1^{\sT{\epsilon}_1} \theta_1^{\,\sD{\epsilon}_1}\cdots  \phi_\ell^{\sT{\epsilon}_\ell} \theta_\ell^{\,\sD{\epsilon}_\ell},
			\end{align}
where
\begin{align}\sT{\epsilon}_i=\begin{cases}1\;\;\text{if}\; \Lambda_i\in \sTD{\Lambda}\,\text{or}\, \sT{\Lambda}\\0\;\;\text{otherwise}\end{cases}
\qquad{\text{and}}\qquad
\sD{\epsilon}_i=\begin{cases}1\;\;\text{if}\; \Lambda_i\in \sTD{\Lambda}\,\text{or}\,\sD{\Lambda}\\0\;\;\text{otherwise.}\end{cases}
\end{align}
						For example we have 
	\begin{gather}
		\Lambda = (\sT{2}, \sD{2}, \sTD{1}, 1 , \sD{0}) \quad \implies\quad
		[\phi;\theta]_{\sL} =		\phi_1\theta_2\phi_3\theta_3\theta_5. 
	\end{gather}
Observe that
\beq \sum_i\sT{\ep}_i=\sT{m}\qquad{\text{and}}\qquad  \sum_i\sD{\ep}_i=\sD{m}\eeq	
so that $\sT{m}$ and $\sD{m}$ are the number of \emph{fermions} of each kind (i.e., the number of $\phi_i$ or $\ta_i$ factors respectively).  Similarly, we will write
\beq x^\La=x_1^{|\La_1|}
x_2^{|\La_2|}
\cdots x_\ell^{|\La_\ell|}.
\label{xla}\eeq These notations allow us to introduce the monomial functions in  a straightforward way.\\
				
		The monomial symmetric functions in superspace, denoted by $m_\Lambda = m_\Lambda(z,\theta,\phi)$, are the superanalogs of the monomial symmetric functions. 
		\begin{definition} \label{def:monomials}
			To every $\Lambda \in \text{SPar}$, we associate the monomial symmetric function
			\begin{align} \label{def. monomial}
				m_\Lambda = \sideset{}{'}\sum_{\omega \,\in\, S_N} \mathcal{K}_\omega  [\phi;\theta]_{\sL}x^\sL,
			\end{align}
			where $\mathcal{K}_\omega$ permutes the indices and the prime indicates that the summation is restricted to distinct permutations of $ [\phi;\theta]_{\sL}x^\sL$. 
		\end{definition}

		\begin{example} Here we give some examples of the monomial functions, the first of which is the superpolynomial $f(x,\theta,\phi)$ of \eqref{eq:exampleSymFunction}. 
			\begin{align}
								m_{(\sT{2}, \sT{0}, \sD{0})}&=	
					\phi_1\phi_2\theta_3(x_1^2 - x_2^2) + 
					\theta_1\phi_2\phi_3(x_2^2-x_3^2) +
					\phi_1\theta_2\phi_3(x_3^2 - x_1^2),\\
								m_{(\sTD{3},\sT{2},2,\sTD{1},\sD{1},\sT{0}, \sD{0})}&= 
					\phi_1\theta_1 \phi_2 \phi_4\theta_4 \theta_5 \phi_6 \theta_7 {x_1^3x_2^2x_3^2x_4x_5}+ \text{distinct permutations} ,\\
								m_{(\sTD{2}, \sT{1}, 1, 1)}&= 
					\phi_1 \theta_1 \phi_2 x_1^2 x_2 x_3 x_4 + \text{distinct permutations}.
			\end{align}
Note that we took the number of variables to be equal to the length of the indexing superpartition.
		\end{example}

																												\begin{proposition}
The monomials $m_\La$, for $\La\in\text{SPar}$, form a basis for
$\mathscr{P}^{S_\infty}(\mathbb{Z})$.
\end{proposition}
\noindent The proof is omitted since it is a straightforward extension of Theorem 16 and Corollary 17 in \cite{classicalN1}. 	\section{Multiplicative bases and their generating functions}

	\subsection{Superpolynomials vs differential forms}

We now borrow from the analysis of the ${\mathcal N}=1$ case \cite{classicalN1} a convenient differential-form formalism. The notion of $({k},{k'})$-form is defined as
\begin{equation}\label{form}f(x)=\sum_{\substack{1\leq i_1,\ldots,i_{{k}}\leq N\\1\leq j_1,\ldots,j_{k'}\leq N}}
f_{ i_1,\ldots,i_{k},j_1,\ldots,j_{k'}}(x)\sT{d}x^{i_1}\wedge\cdots\wedge \sT{d}x^{i_{k}}\wedge \sD{d}x^{j_1}\wedge\cdots\wedge \sD{d}x^{j_{k'}}\,  ,
\end{equation} where the exterior    product $\wedge$ is antisymmetric
and $\sT{d}x^i$ and $\sD{d}x^i$ are considered to be different. 
Therefore, not only 
\beq \sT{d}x^i\wedge\sT{d}x^j=-\sT{d}x^j\wedge\sT{d}x^i\qquad\text{and}\qquad \sD{d}x^i\wedge\sD{d}x^j=-\sD{d}x^j\wedge\sD{d}x^i,\eeq
but also
\beq \sT{d}x^i\wedge\sD{d}x^j=-\sD{d}x^j\wedge\sT{d}x^i.
\eeq\\

The relation with superpolynomials in the $(\sT{m},\sD{m})=(k,k')$ fermionic sector is immediate with the corespondence
\beq\label{COR1}
\sT{d}x_i \longleftrightarrow \phi_i  \qquad\text{and}\qquad 
\sD{d}x_i \longleftrightarrow \theta_i.
\eeq

This relationship also captures the essence of the symmetry requirement imposed on the superpolynomials, namely the invariance with respect to the simultaneous interchange of the two triplets $(x_i,\phi_i,\ta_i)$ and $(x_j,\phi_j,\ta_j)$. Indeed, in this differential-form representation, there is a clear  reason for which $\phi_i$ and $\ta_i$ are entangled with $x_i$ in the symmetry transformations: it is because they are both constructed from $x_i$.
\\

Let us also introduce two distinct exterior
differentiations on forms, $\sT{d}$ and $\sD{d}$:
																						\begin{align}
		 	\sT{d} f(x) &= \sum_{i}	\sT{d} x_{i}\w \partial_{i} f(x)\, \longleftrightarrow \sum_{i}\phi_i	\partial_{i} f(x) \\
		 	\sD{d} f(x) &= \sum_{i}	 \sD{d} x_{i}\w\partial_{i} f(x)\,\longleftrightarrow \sum_{i}	\theta_i\partial_{i} f(x).
		 \end{align}
		From $\sT{d}$ and $\sD{d}$, we construct the following two operators:
\beq \label{eq:bigdiff}
			\sT{\mathcal{D}}  = \sT{d}t \wedge \sT{d} \qquad\text {and}\qquad
			\sD{\mathcal{D}}  = \sD{d}t \wedge \sD{d}.
		\eeq
	where $t$ is an auxiliary parameter (typically a degree-counting variable in a generating function).
$\sT{\mathcal{D}} $ transforms thus a $(k,k')$-form into a $(k+2,k')$-form. Out of these, we construct the master operator:
		\begin{align}
			\mathcal{O} := (1+\sT{\mathcal{D}})(1+\sD{\mathcal{D}}) \label{def:bigdiffO}.
		\end{align}
		This operator will be our main tool for calculations involving generating functions, in whose context we will use the correspondences  
		\beq\label{COR2}
		t\,\sT{d}t \longleftrightarrow\,\sTt  \qquad\text{and}\qquad t\,\sD{d}t\longleftrightarrow\sDt
		\eeq where both $\sTt$ and $\sDt$ are anticommuting variables.\\

		To capture the correspondences \eqref{COR1} and \eqref{COR2}, we will often use the following notation:
\beq\label{NOTA} \ll F(x_i,\sT{d}x_i,\sD{d}x_i,\sT{d}t,\sD{d}t)\rr:=\l[F(x_i,\sT{d}x_i,\sD{d}x_i,\sT{d}t,\sD{d}t)\r]_{\substack{\sT{d}x_i\to\phi_i,\,\sD{d}x_i\to\ta_i\\t\sT{d}t\to\sTt,\,t\sD{d}t\to\sDt}}\eeq
with the understanding that after these substitutions, the wedge products are eliminated.
For instance:
\beq \ll \sT{d}x_2\wedge\sD{d}x_5\wedge t\sD{d}t\rr=\phi_2\,\theta_5\,\sDt.
\eeq\\

\noindent {\bf Key observation:}  The key point of this differential-form formalism is   the following: if $f(x)$ is a symmetric polynomial of degree $n$, by construction, $\ll\sT{d}\,f\rr, \ll\sD{d}\,f\rr$ and $\ll\sT{d}\,\sD{d}\,f\rr$ are guaranteed to be, 
respectively, symmetric superpolynomials of degree $(n-1|1,0), (n-1|0,1)$ and $(n-2|1,1)$. 
{In the case of a multiplicative basis of the ring of symmetric functions, it turns out that applying $\sT{d}, \sD{d}$ and $\sT{d}\, \sD{d}$ on the generators provide natural generators for a $\mathcal N=2$ multiplicative basis. 
Moreover, if $G(x,t)$ is the generating function of the generators of that multiplicative basis, $\ll \O G(x,t)\rr$, being constructed out of the exterior derivatives acting on $G(x,t)$, will be the generating function for its ${\mathcal N}=2$ extension.}

	\subsection{Structure of   the multiplicative bases}

	 We now consider  the extension of the three  multiplicative bases $p_\la, h_\la$ and $e_\la$. Denote a generic basis for ${\mathcal N}={2}$ by $f_\Lambda = f_\Lambda(z,\theta,\phi)$.
 The structure of the multiplicative basis is
 \begin{align}\label{multiA}
 f_\Lambda = \tilde f_{\Lambda_1}\cdots\tilde f_{\Lambda_\ell},\end{align}
 where
\beq\label{multiB}
\tilde f_{\Lambda_i}=\begin{cases}
\sTD{f}_n\;\text{if}\; \Lambda_i\in\sTD{\Lambda}\\
\sT{f}_n\;\text{if}\; \Lambda_i\in\sT{\Lambda}\\\sD{f}_{n}\;\text{if}\; \Lambda_i\in\sD{\Lambda}\\ {f}_{n}\,\;\text{if}\; \Lambda_i\in {\La_0}
\end{cases}\qquad\text{for}\;n={|\Lambda_i|}.
\eeq
Here $f_n$ is the classical (with no anticommuting variables) symmetric function under consideration (either $p_n,e_n$ or $h_n$ -- see below) while 
$\sT{f}_n, \sD{f}_n$ and $\sTD{f}_n$ are respectively
 \beq 
 \sT{f}_n\propto \llbracket\sT{d}f_{n{+1}}\rrbracket,
 \quad  \sD{f}_n\propto \ll\sD{d}f_{n{+1}}\rr,
 \quad  \sTD{f}_n\propto\ll\sT{d} \,\sD{d}f_{n{+2}}\rr,
 \eeq
 with proportionality constants depending on the choice of basis
and where
 we use the notation defined in \eqref{NOTA}.  The expressions for the multiplicative-basis generators are collected in Table \ref{tab1}. That these three multiplicative bases  are indeed genuine bases is demonstrated in Appendix C.\\

\begin{remark} Note that \eqref{multiA} could be written compactly as
$f_\Lambda = \prod_{i=1}^\ell \tilde f_{\Lambda_i}$.
But we stress that the order of multiplication matters (a different order may induce a sign difference).
We take the convention that products are always done with increasing value of the index from left to right, which means for instance that the following is true for every $k$:
\beq
				\prod_{i=1}^{\ell} \tilde f_{\Lambda_i}= \prod_{i=1}^{k} \tilde f_{\Lambda_i} \prod_{j = k + 1}^{\ell}\tilde f_{\Lambda_j}.\eeq
\end{remark}

\begin{table}[h]
   \caption{\label{tab1}The multiplicative bases (with $|\La_i|=n$)}\setlength\extrarowheight{4pt}
\center{\begin{tabular}{|c|l||c|c|c|}
  \hline\hline
$f_\La$&$\tilde f_{\La_i}$&$(\sT{d},\sD{d})$ expression& $m_\La$ expression\\
  \hline
$p_\La$&$p_n$&$p_n$& $m_{(n)}$\\ &$\sT{p}_n$&$\ll \frac{1}{(n+1)}\sTd \,p_{n+1}\rr$& $ m_{(\sT{n})}$\\ &$\sD{p}_n$&$\ll \frac{1}{(n+1)}\sDd \,p_{n+1}\rr$& $ m_{(\sD{n})}$\\ &$\sTD{p}_n$&$\ll \frac{1}{(n+2)(n+1)}\sTd\, \sDd  \,p_{n+2}\rr$& $ m_{(\sTD{n})}$\\ &&&\\  \hline
$h_\La$&$h_n$&$h_n$& $\sum_{\lambda \vdash n } m_\lambda, $\\ &$\sT{h}_n$&$\ll \sTd \,h_{n+1}\rr$& $ \sum_{\Lambda \vdash (n|0,1)} (|\sT{\Lambda}_1|+1) m_\Lambda$\\
&$\sD{h}_n$&$\ll \sDd \,h_{n+1}\rr$& $\sum_{\Lambda \vdash (n|1,0)} (|\sD{\Lambda}_1|+1) m_\Lambda$\\
&$\sTD{h}_n$&$\ll\sTd\, \sDd  \,h_{n+2}\rr$
& $  \sum_{\Lambda \vdash (n|1,1)} \left( c_1\delta_{\sTL,\varnothing} + c_2\delta_{\sTDL,\varnothing} \right)m_\Lambda$
\\&&&\\  \hline
$e_\La$&$e_n$&$e_n$& $m_{(1^n)} $\\ &$\sT{e}_n$&$\ll \sTd \,e_{n+1}\rr$& $ m_{(1^n, \sT{0})}$\\
&$\sD{e}_n$&$\ll \sDd \,e_{n+1}\rr$&$m_{(1^n, \sD{0})}$\\
&$\sTD{e}_n$&$\ll\sTd\, \sDd  \,e_{n+2}\rr$& $ {m_{(1^n, \sT{0},\sD{0})}}$\\
&&&\\ \hline\hline
\end{tabular}}
{\footnotesize \begin{align*}
	\text{where}\quad	c_1=(|\sTD{\Lambda}_1|+2)(|\sTD{\Lambda}_1|+1)\qquad 
			\text{and} \qquad	c_2 = \left\{ 
		\begin{array}{ll} 
			\phantom{-}(|\sTL_1|+1)(|\sDL_1|+1)	& 	\text{if } |\sTL_1| \geq |\sDL_1| \\
			-(|\sTL_1|+1)(|\sDL_1|+1)	&	\text{otherwise} 
		\end{array} 
		\right. 
\end{align*}}
\end{table}

	\begin{example}
			Let $f(x,\theta,\phi)$ be the symmetric superpolynomial \eqref{eq:exampleSymFunction}. We then have 
			\begin{align}
				f(x,\theta,\phi) &= 
																																																								p_{\, \superYsmall{
				\,& \,& \yT\\
				\yT\\
				\yC}}
				+p_{\, \superYsmall{
				\,& \,& \yTC\\
				\yT}}
				-p_{\, \superYsmall{
				\,& \,& \yT\\
				\yTC}} 			\nonumber\\
				&= 
				-\frac{8}{3}h_{\, \superYsmall{\,& \yT\\
				\,\\
				\yT\\
				\yC}}
				-\frac{1}{3}h_{\, \superYsmall{\,& \yT\\
				\,& \yC\\
				\yT}}
				-\frac{2}{3}h_{\, \superYsmall{\,\\
				\,\\
				\yTC\\
				\yT}}
				-\frac{1}{3}h_{\, \superYsmall{\,& \yTC\\
				\,\\
				\yT}}
				+h_{\, \superYsmall{\,& \yT\\
				\,\\
				\yTC}}
				+\frac{7}{3}h_{\, \superYsmall{\,& \,& \yT\\
				\yT\\
				\yC}}
				+\frac{2}{3}h_{\, \superYsmall{\,& \,\\
				\yTC\\
				\yT}}
				+\frac{1}{3}h_{\, \superYsmall{\,& \,& \yTC\\
				\yT}}
				-h_{\, \superYsmall{\,& \,& \yT\\
				\yTC}}
							\nonumber\\& = 
				\phantom{-}\frac{2}{3}e_{\, \superYsmall{\,& \yT\\
				\,\\
				\yT\\
				\yC}}
				+\frac{1}{3}e_{\, \superYsmall{\,& \yT\\
				\,& \yC\\
				\yT}}
				+\frac{2}{3}e_{\, \superYsmall{\,\\
				\,\\
				\yTC\\
				\yT}}
				+\frac{1}{3}e_{\, \superYsmall{\,& \yTC\\
				\,\\
				\yT}}
				-e_{\, \superYsmall{\,& \yT\\
				\,\\
				\yTC}}
				-\frac{1}{3}e_{\, \superYsmall{\,& \,& \yT\\
				\yT\\
				\yC}}
				-\frac{2}{3}e_{\, \superYsmall{\,& \,\\
				\yTC\\
				\yT}}
				-\frac{1}{3}e_{\, \superYsmall{\,& \,& \yTC\\
				\yT}}
				+e_{\, \superYsmall{\,& \,& \yT\\
				\yTC}}.
			\end{align}
		\end{example}

\begin{remark}The particular normalization used for the power-sum generators is such that the monomial expansion is monic (here $n\geq 0$ except for $p_n$ where $n\geq1$):			\begin{align}
	p_n	&= \sum_{k} x_k^n = m_{(n)},				& 	\sT{p}_n &= \sum_{k} \phi_k x_k^n = m_{(\sT{n})}, \\
				\sD{p}_n	&= \sum_{k} \theta_k x_k^n = m_{(\sD{n})},			&\sTD{p}_n	&= \sum_{k} \phi_k\theta_k x_k^n = m_{(\sTD{n})}.
			\end{align}
\end{remark}
		\begin{example}
			If $\Lambda = (\sT{2}, \sTD{1}, \sTD{1}, \sD{0})$ then the power sum associated to $\Lambda$ is 
			\begin{align}
				p_{ \superYsmall{
				\,&\,&\yT \\
				\,&\yTC \\
				\,&\yTC \\
				\yC
				}}
				= (\phi_1 x_1^2 + \phi_2 x_2^2 + \cdots)(\phi_1 \theta_1 x_1 + \phi_2\theta_2 x_2 + \cdots)^2 (\theta_1 + \theta_2 + \cdots).
			\end{align}
		\end{example}
		
	\begin{remark}
		The prefactors of the monomial expansion of    
		$\sT{h}_n, \sD{h}_n, \sTD{h}_n $ are related 
to the action of the exterior differentiations on the monomial basis. For instance, the action of $\sT{d}$ on  a monomial with no fermionic part is:
		\begin{align}
			\sTd m_\lambda  = \sideset{}{'}\sum_{\lambda_k \in \lambda} \lambda_k \,m_{(\lambda_1, \cdots, \sT{\lambda_k - 1}, \cdots)}
		\end{align}
		where the prime on the summation indicates that we sum only {on parts of $\lambda$ associated to rows ending with a removable corner (i.e., a box whose removal still yields the diagram of  a partition).} On the diagram, this operation amounts to replacing a box by a v-circle in every distinct way and adding a factor given by the horizontal position of the v-circle. As an illustration of the above formula, consider
		\begin{align}
			\sTd m_{\superYsmall{
				\,&\,&\, \\
				\,&\,&\, \\
				\,&\, \\
				\,
				}}
				\longleftrightarrow 
				3m_{\superYsmall{
				\,&\,&\, \\
				\,&\,&\yT \\
				\,&\, \\
				\,
				}}
				+
				2m_{\superYsmall{
				\,&\,&\, \\
				\,&\,&\, \\
				\,&\yT \\
				\,
				}}
				+
				m_{\superYsmall{
				\,&\,&\, \\
				\,&\,&\, \\
				\,&\, \\
				\yT
				}}.
		\end{align}
		\end{remark}
	
		\begin{example} Let $\Lambda = (\sT{2}, \sTD{1}, 1)$, then $h_\Lambda$ is
			\begin{align}
				h_{\superYsmall{
				\, &\, &\yT\\
				\,&\yC &\yT\\
				\,
				}}
				=
				\left( 	3 m_{ \superYsmall{\,&\,&\yT}}+ 
						m_{\superYsmall{\,&\,\\ \yT}}+
						2 m_{\superYsmall{\,&\yT\\ \,}}+
						m_{\superYsmall{\,\\ \, \\ \yT}}
				\right)
				\left( 	3 \cdot 2 m_{ \superYsmall{\,&\yC&\yT}} +
						2 m_{\superYsmall{\,& \yT \\ \yC}} - 2 m_{\superYsmall{\,& \yC \\ \yT}}+
						2 m_{\superYsmall{\,\\ \yC & \yT}} +  m_{\superYsmall{\,\\ \yC \\ \yT}}
				\right)
				\left( m_{\superYsmall{\,}} \right).
			\end{align} 
			Note that we wrote the vh-circle as two circles in order to stress the origin of the numerical factors.
		\end{example}

	\subsection{Generating functions}

		Using the notation \eqref{NOTA}, we have the following identity:\beq\label{IGF}\ll(1+\overline\D)(1+\underline \D)F(tx)\rr=F(tx+\sTt\phi+\sDt\ta).\eeq
		The proof is elementary and is relegated to Appendix \ref{apendix:ProofIdent}.
		The following proposition is an immediate consequence of {the key observation (cf. the end of Section 3.1) and the previous identity.}
						\begin{proposition} If we denote by $G(t)=\sum_{n\geq 0}t^nf_n$ the generating function for the multiplicative basis $f_\la$, the generating function for the ${\mathcal N}=2$ families of generators is 
			\beq
	{G}(t,\sTt,\sDt) =\ll (1+\overline\D)(1+\underline \D) G(t)\rr {=G(tx+\sTt\phi+\sDt\ta)}.\eeq	
		\end{proposition}		
				 In particular, with the convention that $p_0=0$ and $h_0=e_0=1$, the three generating functions are
		\begin{align} \label{eq:powersumgenfct}
			{P}(t,\sTt,\sDt) &= \sum_i \frac{t x_i+\sTt\phi_i+\sDt\theta_i}{1-(t x_i+\sTt\phi_i+\sDt\theta_i)} \nonumber \\
			&=\sum_{n\geq0} t^n(p_n +(n+1)\sTt\,\sT{p}_n + (n+1)\sDt\,\sD{p}_n -(n+1)(n+2)\sTt\,\sDt\,\sTD{p}_n),
		\\
				{H}(t,\sTt, \sDt) &= \prod_i \frac{1}{1-(tx_i+\sTt\phi_i +\sDt\theta_i)} = \sum_n t^n( h_n +\sTt\,\sT{h}_n +\sDt\,\sD{h}_n - \sTt\,\sDt\,\sTD{h}_n), \label{eq:GenComplete}\\
			{E}(t,\sTt, \sDt)	&= \prod_i (1+tx_i+\sTt\phi_i+\sDt\theta_i) = \sum_n t^n (e_n +\sTt\, \sT{e}_n +\sDt\, \sD{e}_n - \sTt\,\sDt \,\sTD{e}_n). \label{eq:genelementary}
			\end{align}
{We stress that in the above expressions, the coefficients of $t^n$ are all bosonic (i.e., commuting) since $\sTt$ and $\sDt$ are, we recall, Grassmannian variables.} 
Observe that the signs in the terms involving $\sTt\,\sDt$ come from the anticommutation of $\overline \D$ and $\sDt$. 

	\subsection{Relations between bases}		The expressions \eqref{eq:powersumgenfct}, \eqref{eq:GenComplete} and \eqref{eq:genelementary} are related as follows:
		\begin{align}
				H(t,\sTt,\sDt)	 E(-t,-\sTt,-\sDt) &= 1,\label{HE}\\
			{H}(t,\sTt,\sDt) {P}(t,\sTt,\sDt) &= (t\partial_t + \sDt\partial_{\sDt} + \sTt\partial_{\sTt}){H}(t,\sTt, \sDt),\label{eq:HPequal} \\
				E(t,\sTt,\sDt)	 P(-t,-\sTt,-\sDt) &= -(t\partial_t + \sDt \partial_{\sDt} +\sTt \partial_{\sTt})	 E(t,\sTt,\sDt) .\label{eq:EPequal}
		\end{align} 
		From these relations, we derive a number of identities displayed in Table \ref{tab2}.\\

		\begin{table}[h]
		   \caption{\label{tab2}Relations between the multiplicative bases}
		\setlength\extrarowheight{6pt}
		\center{\begin{tabular}{|c|l||c|}
		  \hline\hline
		Related bases&Identities and recursion relations\\
		  \hline
		A: $\;e_n$ vs $h_n$&$\sum_{r=0}^n (-1)^r e_r\, h_{n-r} = 0$\\
		&$\sum_{r=0}^{n}(-1)^r(e_r\, \sT{h}_{n-r} - \sT{e}_r\, h_{n-r}) = 0$\\
		&$\sum_{r=0}^{n} (-1)^r(\sTD{e}_r\, h_{n-r} - \sT{h}_{n-r}\, \sD{e}_r- \sT{e}_r \, \sD{h}_{n-r}+e_r \, \sTD{h}_{n-r})=0$\\ &\\\hline\hline
		B: $\;h_n$ vs $p_n$&$nh_n = \sum_{r=1}^n p_r\, h_{n-r}$\\
		&$(n+1)\sT{h}_n = \sum_{r=0}^n(p_r \, \sT{h}_{n-r} + (r+1)\sT{p}_r \, h_{n-r})$\\
		&$(n+2)\sTD{h}_n = \sum_{r=0}^n( p_r\, \sTD{h}_{n-r} + (r+1)(\sT{p}_r\sD{h}_{n-r} + \sT{h}_{n-r}\sD{p}_r)+(r+2)(r+1)\sTD{p}_r\, h_{n-r})$\\ &\\\hline\hline
		C: $\;e_n$ vs $p_n$&$ne_n = \sum_{r=1}^n (-1)^{r+1}p_re_{n-r}$\\
		&$(n+1)\sT{e}_n = \sum_{r=0}^n(-1)^{r+1}(p_r \, \sT{e}_{n-r}-(r+1)\sT{p}_r \, e_{n-r})$\\
		&$(n+2)\sTD{e}_n = \sum_{r=0}^n(-1)^{r+1}(p_r\, \sTD{e}_{n-r}  - (r+1)(\sT{p}_r\, \sD{e}_{n-r} + \sT{e}_{n-r}\, \sD{p}_r)+ (r+2)(r+1)\sTD{p}_r e_{n-r})$\\ &\\\hline\hline
		\end{tabular}}
		\end{table}

		The first identity  in each group is the classical one \cite{Macdonald1998}. The middle one (and its equivalent form obtained by replacing $\o{f}_n\to\u{f}_n$) has been derived in \cite{classicalN1}. The third one is novel.  Note that  the second and third relations in each group can also be recovered by operating on the first one with $\sT{d}$ and $\sT{d}\, \sD{d}$ respectively, using the differential-form correspondences given in Table \ref{tab1}.\\

		 Let us introduce the homomorphism $\hat{\omega}:\mathscr{P}^{S_\infty} \longrightarrow \mathscr{P}^{S_\infty}$  defined in the following way on the generators of the elementary basis: 
		 \begin{align}
		 	\hat{\omega} : e_n \longrightarrow h_n,\quad \sT{e}_n \longrightarrow \sT{h}_n,\quad \sD{e}_n \longrightarrow \sD{h}_n,\quad \sTD{e}_n \longrightarrow \sTD{h}_n .\label{eq:involution}
		 \end{align}
		 \begin{proposition}
		 	The homomorphism $\hat{\omega}$ is an involution ($\hat{\omega}^2=1$), which implies that
		 	\begin{align} 
		 		\hat{\omega} : h_n \longrightarrow e_n,\quad \sT{h}_n \longrightarrow \sT{e}_n,\quad \sD{h}_n \longrightarrow \sD{e}_n,\quad \sTD{h}_n \longrightarrow \sTD{e}_n.
		 	\end{align}
		 \end{proposition}
		 \begin{proof}{ The proofs of these four relations are similar. Since 
		the cases $h_n$ and $\sT{h}_n$ (and hence also $\sD{h}_n$) are worked out in \cite{Macdonald1998} and \cite{classicalN1} respectively, we only focus on $\sTD{h}_n$.} We proceed by induction. The case $n=0$ follows easily from the action of $\hat \omega$ on A-3 (the third relation in Block A  of Table \ref{tab2}) with $r=n=0$.  We then suppose
that $\hat \omega(\sTD{h}_r)= \sTD{e}_r$ for all $r<n$.  
Acting with $\hat{\omega}$ 
		 on A-3 gives
		 \begin{align}
		 	0=&\sum_{r=0}^{n} (-1)^r\bigr(\hat{\omega}(\sTD{e}_r) \hat{\omega}(h_{n-r}) 
		 		- \hat{\omega}(\sT{h}_{n-r}) \hat{\omega}(\sD{e}_r)
		 		-\hat{\omega}(\sT{e}_r) \hat{\omega}(\sD{h}_{n-r})
		 		+\hat{\omega}(e_r) \hat{\omega}(\sTD{h}_{n-r})\bigl) \nonumber\\
		 		=&\sum_{r=0}^{n} (-1)^r\bigl( \sTD{h}_r \hat{\omega}(h_{n-r}) - 
		 		\hat{\omega}(\sT{h}_{n-r}) \sD{h}_r-\sT{h}_r \hat{\omega}(\sD{h}_{n-r}) +h_r\hat{\omega}(\sTD{h}_{n-r}) \bigr)\nonumber\\
		 		=&(-1)^n\sum_{r=0}^{n} (-1)^r\bigl( \sTD{h}_{n-r} \hat{\omega}(h_{r}) - 
		 		\hat{\omega}(\sT{h}_{r}) \sD{h}_{n-r}-\sT{h}_{n-r} \hat{\omega}(\sD{h}_{r}) +h_{n-r}\hat{\omega}(\sTD{h}_{r}) \bigr)\nonumber \\
		 		=&(-1)^n\sum_{r=0}^{n} (-1)^r(e_r \,\sTD{h}_{n-r} -\sT{e}_r \,\sD{h}_{n-r} - \sT{h}_{n-r}\, \sD{e}_r +\hat \omega(\sTD{h}_r)\, h_{n-r}).
		 \end{align}
Comparing with A-3 and using the induction hypothesis, we immediately have that
$\hat{\omega}(\sTD{h}_n) = \sTD{e}_n$.
{(The relations for $h_n$ and $\sT{h}_n$ similarly rely on A-1 and A-2 respectively.)}		 \end{proof}

		 Given this, we are now in position to obtain the action of the homomorphism on the power sum basis. 
		 \begin{proposition}
		 	Let $\hat{\omega}$ be the involution defined in \eqref{eq:involution}.  We then have 
		 	\begin{align}
		 		\omegahat :\quad &p_n \longrightarrow (-1)^{n-1}p_n, \!\!\!\!\!\!&&\sT{p}_n \longrightarrow (-1)^{n}\sT{p}_n, \\
&		 		 \sD{p}_n \longrightarrow (-1)^{n}\sD{p}_n, \!\!\!\!\!\! &&\sTD{p}_n \longrightarrow (-1)^{n-1}\sTD{p}_n,
		 	\end{align}
		 	so that the resulting action on $p_\Lambda$ is 
		 	\begin{align}
		 		\omegahat(p_\Lambda) = \omega_\Lambda p_\Lambda, \quad \omega_\Lambda \equiv (-1)^{|\Lambda| - \ell(\sTDL) - \ell(\Lambda_0)}.
		 	\end{align}
		 \end{proposition}
						
		 \begin{proof} The proof follows directly by applying $\omegahat$ on equations B-2 and B-3 of Table \ref{tab2} 	and comparing with equations C-2 and C-3.	 \end{proof}

	\subsection{Scalar product and reproducing kernel}
				\begin{definition} We define a scalar product on $\mathscr{P}^{S_\infty}$ by
\beq
			\label{PS1}	\langle{ p_\Lambda^\top} | p_\Omega \rangle = z_\Lambda \zeta_\Lambda \delta_{\Lambda,\Omega}, \eeq
				where			\begin{align}\label{zdef}
				z_\Lambda &= z_{\Lambda_0} = \prod_{k\geq 1} k^{n_{\Lambda_0}(k)}n_{\Lambda_0}(k)! , \\
				\zeta_\Lambda &= \zeta_{\sTDL}= \prod_{k\geq 0} \frac{n_{\sTDL}(k)!}{(k+1)^{n_{\sTDL}(k)}}.\label{zetadef}
			\end{align}
			and \beq p_\La^\top=(-1)^{\binom{\o m+\u m}{2}}p_\La,\eeq 
with  $n_\lambda(i)$ denoting the number of parts equal to $i$ in the partition $\lambda$.
 
		\end{definition}
		The sign is only introduced to simplify the computations that will appear later on.  Since the sign does not vary in a given fermionic sector, the norm of any element in that sector is either always positive or always negative. 
		Observe that the symbol $\top$ stands for the operation that reverses the order of the anticommuting variables.\\

				The norm of $p_\Lambda$ in the scalar product comes from the reproducing kernel that we now introduce.

		\begin{theorem} \label{Th: Kernel} Let $K(x,y,\phi, \sTt, \theta, \sDt)$ be the bi-symmetric formal power series given by
			\begin{align}
				K(x,y,\phi, \sTt, \theta, \sDt) = \prod_{i,j}\frac{1}{1-x_iy_j -\phi_i\sTt_j -\theta_i\sDt_j} .
			\end{align}
			We then have
			\begin{align}
				K(x,y,\phi,\theta,\sTt,\sDt)
				&= \sum_{ \Lambda \,\in\, \text{SPar} } (-1)^{\binom{\sT{m}+\sD{m}}{2}} z_\Lambda^{-1} \zeta_\Lambda^{-1} p_\Lambda(x,\phi,\theta)p_\Lambda(y,\sTt,\sDt).
			\end{align}
		\end{theorem}

\begin{remark} The reproducing kernel is the natural generalization of the cases $\Nm=0$ and $\Nm=1$.
But the actual justification is that this expression leads to a scalar product with respect to which $m_\La$ and $h_\Om$ are dual to each other (see Proposition \ref{mhdual} below).
\end{remark}

		\begin{proof} We have
			\begin{align}
				\prod_{i,j}&\frac{1}{1-x_iy_j - \phi_i\sTt_j - \theta_i\sDt_j} \nonumber \\
					&= \exp\l(\sum_{i,j}\ln(1-x_iy_j - \phi_i\sTt_j - \theta_i\sDt_j)^{-1}\r)	\nonumber \\
					&=\exp\left\lbrace\sum_{i,j} \sum_{n\geq 1}\frac{1}{n}(x_iy_j + \phi_i\sTt_j + \theta_i\sDt_j)^n )\right\rbrace \nonumber \\
										&=\exp\left\lbrace\sum_{i,j} \sum_{n\geq 1}\frac{1}{n}\left( (x_iy_j)^n +\sum_{k\geq 1}^{n}{\binom{n}{k}}(x_iy_j)^{n-k} (\phi_i\sTt_j + \theta_i\sDt_j)^k \right) \right\rbrace \nonumber \\
					&=\exp\left\lbrace\sum_{i,j} \sum_{n\geq 1}\frac{1}{n}\left( (x_iy_j)^n +\frac{n!}{(n-1)!1!}(x_iy_j)^{n-1} (\phi_i\sTt_j + \theta_i\sDt_j)\right. \right.\nonumber \\
					& \qquad \quad \left.\left. + \frac{n!}{(n-2)!2!}(x_iy_j)^{n-2} (\phi_i\sTt_j + \theta_i\sDt_j)^2 \right) \right\rbrace \nonumber \\
					&= \exp\left\lbrace \sum_{n \geq 1} \frac{p_n(x)p_n(y)}{n}+\sum_{n\geq 0}\bigl(\sT{p}_n(x,\phi)\sT{p}_n(y,\sTt) +\sD{p}_n(x,\theta)\sD{p}_n(y,\sDt)\bigr) \right. \nonumber \\
					&\qquad \quad \left. - \sum_{n\geq 0}(n+1)\sTD{p}_n(x_i, \phi_i, \theta_i) \sTD{p}_n(y_j,\sTt_j,\sDt_j)   \right\rbrace.
			\end{align}
			In the fourth equality,  we used the fact that $(\phi_i\sTt_j + \theta_i\sDt_j)^3=0$.  Note also that
the term in 
$(\phi_i\sTt_j + \theta_i\sDt_j)^2$ only appears if $n\geq 2$ in the penultimate
expression.   
Now,
			using Theorem 33 of \cite{classicalN1}, we can write
			\begin{align}
				&\exp\left\lbrace \sum_{n \geq 1} \frac{p_n(x)p_n(y)}{n}+\sum_{n\geq 0}(\sT{p}_n(x,\phi)\sT{p}_n(y,\sTt) +\sD{p}_n(x,\theta)\sD{p}_n(y,\sDt))\right\rbrace \nonumber\\
				&={	\sum_{\substack{\lambda \,\in\, \text{Par} \\ \mu,\nu \,\in\, \text{ParD}}}  }
				z_\lambda^{-1}p_\lambda(x)p_\lambda(y)
				{(-1)^{\binom{\ell(\mu)}{2}+\binom{\ell(\nu)}{2}}}\,\sT{p}_\mu(x,\phi)\sT{p}_\mu(y,\sTt)\sD{p}_\nu(x,\theta)\sD{p}_\nu(y,\sDt),
			\end{align}
			where  Par and ParD stand respectively for the set of partitions and the set of partitions with distinct parts.
We can thus concentrate on the remaining term:
			\begin{align}
				&\prod_{n\geq0} \exp \left\lbrace -(n+1)\sTD{p}_n(x,\phi,\theta)\sTD{p}_n(y,\sTt,\sDt)  \right\rbrace \nonumber \\
				&\quad= \prod_{n\geq0}\sum_{k\geq0}\frac{1}{k!} \left[ -(n+1)\sTD{p}_n(x,\phi,\theta)\sTD{p}_n(y,\sTt,\sDt)  \right]^k \nonumber \\
				&\quad= \prod_{n\geq0}\sum_{k\geq0} \frac{(-1)^k (n+1)^k}{k!} \left[  \sTD{p}_n(x,\phi,\theta)\sTD{p}_n(y,\sTt,\sDt)  \right]^k \nonumber \\
				&\quad= \sum_{\rho \,\in\, \text{Par}_0} (-1)^{\ell(\rho)}\zeta_\rho^{-1} \sTD{p}_\rho(x,\phi,\theta)\,\sTD{p}_\rho(y,\sTt,\sDt),
			\end{align} 
			where $\text{Par}_{0}$ is the set of ordinary partitions for which we keep track of the number of zeros (we distinguish for instance the partitions $(1)$, $(1,0)$ and $(1,0,0)$, and we also include $(0)$, $(0,0)$, etc.) and where we recall that  
			\begin{align}
				\zeta_\lambda &= \prod_{k\geq 0} \frac{n_\lambda(k)!}{(k+1)^{n_\lambda(k)}}. 
			\end{align}
			Observing that
			\begin{align} 
				\sT{p}_\mu(y,\phi) \sD{p}_\nu(x,\theta) = (-1)^{\ell(\mu)\ell(\nu)}\sD{p}_\nu(x,\theta) \sT{p}_\mu(y,\phi),
			\end{align}
			we combine these results to get 
			\begin{align}
				K(x,y,\phi,\theta,\sTt,\sDt)
				&= \sum_{ \Lambda \,\in\, \text{SPar}} (-1)^{\binom{\ell(\sTL)+{\ell({\sDL})}}{2} + \ell(\sTDL)} z_\Lambda^{-1} \zeta_\Lambda^{-1} p_\Lambda(x,\phi,\theta)p_\Lambda(y,\sTt,\sDt).
			\end{align}
			Using the relation $\ell({\sTL}) + \ell({\sDL}) = \sT{m} + \sD{m} -2\ell(\sTDL)$ we can obtain a more elegant form for the sign.  Denoting  $\ell(\sTDL)$ as $\sTD{\ell}$ for convenience, we have
			\begin{align}
				\frac{(\ell(\sTL)+\ell(\sDL))(\ell(\sTL)+\ell(\sDL)-1)}{2} + \sTD{\ell}
				&= \frac{(\sT{m} + \sD{m} - 2\sTD{\ell})(\sT{m} + \sD{m} - 2\sTD{\ell} -1)}{2}+ \sTD{\ell} \nonumber \\
				&= \frac{(\sT{m} + \sD{m})(\sT{m} + \sD{m} -1)}{2} -\sTD{\ell}\,(2\sT{m} + 2\sD{m} -1) + \sTD{\ell} \; \mod 2 \nonumber\\
				&= \binom{\sT{m} + \sD{m}}{2}\; \mod 2, \label{eq92}
			\end{align}
			so that 
			\begin{align}
				K(x,y,\phi,\theta,\sTt,\sDt)
				&= \sum_{ \Lambda \,\in\, \text{SPar}} (-1)^{\binom{\sT{m} + \sD{m}}{2}} z_\Lambda^{-1} \zeta_\Lambda^{-1} p_\Lambda(x,\phi,\theta)p_\Lambda(y,\sTt,\sDt) \, ,
			\end{align}
			which completes the proof. 
		\end{proof}

		\begin{corollary} $K(x,y,\phi,\theta,\sTt,\sDt)$ is a reproducing kernel in the space of symmetric superfunctions:
			\begin{align}
				\braket{K^\top(x,y,\phi,\theta,\sTt,\sDt) | f(x,\phi,\theta)} = f(y,\sTt,\sDt),\; \forall \; f \in \mathscr{P}^{S_\infty} 
			\end{align}
		{where $\top$ acts only on the variables $\phi, \ta$.}
		\end{corollary}
				\begin{proof} If $f \in \mathscr{P}^{S_\infty}$ then there exist unique coefficients $f_\Lambda$ such that $f = \sum_\Lambda f_\Lambda p_\Lambda$. Hence, 
			\begin{align}
				\braket{K^\top(x,y,\phi,\theta,\sTt,\sDt) | f(x,\theta,\phi)} 
				&= \sum_{\Lambda, \Omega}
								z^{-1}_\Omega \zeta_\Omega^{-1} f_\Lambda \braket{p^\top_\Omega(x,\phi,\theta) | p_\Lambda(x,\phi,\theta)} p_\Omega(y,\sTt,\sDt) \nonumber \\
					&=\sum_\Lambda f_\Lambda p_\Lambda(y,\sTt,\sDt) = f(y, \sTt,\sDt).
			\end{align}
				\end{proof}
		\begin{proposition}
			Let $u_\Lambda$ and $v_\Lambda$ be two bases of $\mathcal{P}^{S_\infty}_{(n|\sD{m},\sT{m})}$.  We then  have
			\begin{align}
				K(x,y,\phi, \sTt, \theta, \sDt) = \sum_{{\rm SPar}} (-1)^{\binom{\sT{m}+\sD{m} }{2} }u_\Lambda(x,\phi,\theta) v_\Lambda(y,\sTt,\sDt) 
				\quad \Longleftrightarrow \quad
				\braket{u_\Lambda^\top | v_\Omega} = \delta_{\Omega \Lambda}.
			\end{align}
		\end{proposition}
		\begin{proof} The proof is similar to that of the cases $\mathcal N=0$
and $\mathcal N=1$ (see \cite{Macdonald1998,classicalN1}) and will thus be omitted. \end{proof}

		\begin{proposition} \label{mhdual}The monomial and the homogeneous symmetric functions are dual bases, that is,
			\begin{align}
				K(x,y,\phi, \sTt, \theta, \sDt) = \sum_{\text{SPar}} (-1)^{\binom{\sT{m}+\sD{m} }{2} }m_\Lambda(x,\phi,\theta) h_\Lambda(y,\sTt,\sDt),
			\end{align}
			or, equivalently, 
			\begin{align} \braket{m_\Lambda^\top | h_\Omega} = \delta_{\Omega \Lambda} .\end{align}
		\end{proposition}
		\begin{proof}
			\begin{align}
				K(x,y,\phi, \sTt, \theta, \sDt) 	&= \prod_i \frac{1}{1-x_i y_i - \phi_i\sTt_i - \theta_i\sDt_i} =\prod_i  H(x_i, \theta_i, \phi_i) \nonumber\\
					&= \prod_i \sum_n x_i^n(h_n + \theta_i\, \sD{h}_n + \phi_i \,\sT{h}_n - \phi_i \theta_i \,\sTD{h}_n) \nonumber \\
					&= \sum_{\alpha} (-1)^{ \binom{\sT{m}+\sD{m} }{2} } [\phi;\theta]_\alpha x^\alpha h_\alpha \nonumber \\
					&= \sum_{\La\in \text{SPar}} (-1)^{ \binom{\sT{m}+\sD{m} }{2} }m_\sL(x,\phi,\theta) h_\sL(y,\sTt,\sDt),
			\end{align}
with $x^\a$ defined as in \eqref{xla} and
where $\alpha$ runs over all sequences of marked nonnegative integers (that is, either overlined, underlined, bilined or ordinary integers) with a finite
number of non-zero entries (considering that $\overline 0$ is not a zero entry for instance).   The sign is understood in the following way.  
Each $h_{\alpha_i}$ (we use the obvious notation $h_{\bar n}$ for $\overline{h}_n$, $h_{\underline{n}}$ for $\underline{h}_n$, and
 $h_{\overline{\underline{n}}}$ for $\overline{\underline{h}}_n$) needs to commute with $[\phi;\theta]_{(\alpha_{i+1}, \alpha_{i+2},\dots)}$ thus producing, when considering all $i$'s, 
a sign of $(-1)^{\binom{\ell(\overline{\alpha})+{\ell(\underline{\alpha})}}{2}}$.  
Taking into account the extra $(-1)^{\ell(\underline{\overline{\alpha}})}$ coming from the terms of the type $-\phi_i \theta_i \,\sTD{h}_n$, we get a total sign of  $(-1)^{\binom{\ell(\overline{\alpha})+{\ell(\underline{\alpha})}}{2} + \ell(\underline{\overline{\alpha}})}$ which, as we have seen in \eqref{eq92}, is equal to $(-1)^{ \binom{\sT{m}+\sD{m} }{2} }$.

		\end{proof}

		We complete this section by taking advantage of the notation introduced in the definition of the scalar product in order to present concise expressions for the complete and elementary symmetric functions in terms of power sums.
		
		\begin{proposition} \label{cor:basisInpowersum}
						The relations between $h$ and $p$ (resp. $e$ and $p$) in Table \ref{tab2} can be rewritten compactly as
			\begin{align}
				h_n &= \sum_{\Lambda \vdash (n|0,0)} z^{-1}_\Lambda p_\Lambda, & e_n &= \sum_{\Lambda \vdash (n|0,0)} z^{-1}_\Lambda \omega_\Lambda p_\Lambda,  \\
				\sT{h}_n &= \sum_{\Lambda \vdash (n|1,0)} z^{-1}_\Lambda p_\Lambda, & \sT{e}_n &= \sum_{\Lambda \vdash (n|1,0)} z^{-1}_\Lambda \omega_\Lambda p_\Lambda, \label{eq:hEnP}\\
			\sTD{h}_n &= \sum_{\Lambda \vdash (n | 1,1)} \zeta_\Lambda^{-1} z_\Lambda^{-1} p_\Lambda, 
			&	\sTD{e}_n &= \sum_{\Lambda \vdash (n | 1,1)} \zeta_\Lambda^{-1} z_\Lambda^{-1} \omega_\Lambda p_\Lambda. \label{eq:newhebarbarenP}\end{align}
Note that the equations on the second line also hold with overlines replaced by underlines and $(n|1,0)$ replaced by $(n|0,1)$.\\
		\end{proposition}
		\begin{proof} Clearly, it suffices to prove the relations  in the first column since the other ones can be obtained by applying $\hat\om$.
			Even though \eqref{eq:hEnP} is proven in \cite[Corollary 36]{classicalN1}, we provide here a new and more direct proof whose ideas will be useful in the demonstration of \eqref{eq:newhebarbarenP}.\\

			We know from \cite{Macdonald1998} that $h_n = \sum_{\lambda \vdash n} z^{-1}_\lambda p_\lambda$.  Acting on both sides with $\sT{d}$, we get 
			\begin{align}
				\sT{h}_{n-1}=\ll\sT{d}h_n\rr =\l\ll\sum_{\lambda \vdash n}z^{-1}_\lambda \sT{d}p_\lambda\r\rr.
			\end{align}
			We calculate the action of the exterior derivative on the power sum
			\begin{align}
				\ll\sT{d}p_\lambda\rr = \sum_{i} \lambda_i\, p_{(\lambda_1, \dots, \overline{\lambda_i-1}, \dots)} 
				= \sideset{}{'}\sum_{i} \lambda_i\, n_\lambda(\lambda_i) \,p_{(\lambda_1, \dots, \overline{\lambda_i-1}, \dots)},
			\end{align}
			where the prime indicates that the sum is restricted {to parts of $\lambda$ associated to rows ending with a removable corner} and where, as usual, $n_\lambda(\lambda_i)$ is the number of occurrence of $\lambda_i$ in $\lambda$. Therefore
			\begin{align}
				\sT{h}_{n-1} = \sum_{\lambda \vdash n}z^{-1}_\lambda  \sideset{}{'}\sum_{i} \lambda_i \,n_\lambda(\lambda_i) \, p_{(\lambda_1, \cdots, \overline{\lambda_i-1}, \cdots)}.
			\end{align}
			It is easy to check that
			\begin{align}
				\frac{\lambda_i n_\lambda(\lambda_i)}{z_\lambda} 
= z_{\lambda^\prime}^{-1}, \quad \text{with } \lambda^\prime = \left. \lambda\right|_{n_\lambda(\lambda_i) \rightarrow n_\lambda(\lambda_i) -1}.
			\end{align}
Hence, we obtain
			\begin{align}
				\sT{h}_{n-1} &= \sum_{\lambda \vdash n} \sideset{}{'}\sum_{i} z^{-1}_{\lambda'}   p_{\lambda'} \sT{p}_{\lambda_i -1} \nonumber \\
				&= \sum_{\Lambda \vdash (n-1|1, 0)} z^{-1}_\Lambda p_\Lambda,
			\end{align}
which completes the proof of the first relation in \eqref{eq:hEnP}. 						We  can then obtain the first relation in \eqref{eq:newhebarbarenP} by acting on $\sD{h}_n$ with $\sT{d}$:
			\begin{align} 
			\sTD{h}_{n-1}=\ll\sT{d}\sD{h}_n \rr&= \l\ll\sum_{\Lambda \vdash (n|1,0)}z^{-1}_\Lambda \sT{d}p_\Lambda\r\rr \nb\\
				&= \sum_{\Lambda \vdash (n|1,0)} z^{-1}_\Lambda
				\left[\, \sum_{{i ,\, \Lambda_i \notin \sD{\Lambda}}} n_\Lambda(\Lambda_i)\,\Lambda_i\, p_{(\Lambda_1, \cdots, \overline{\Lambda_i - 1}, \cdots)} + \Lambda_k \, p_{(\Lambda_1, \cdots, \sTD{\Lambda_k -1}, \cdots)} \right], \label{eq104} 
			\end{align}
where $k$ is the position of the overlined entry in $\Lambda$.  As in the first part of the proof, we have
\begin{align} \label{eq105}
 \sum_{\Lambda \vdash (n|1,0)} z^{-1}_\Lambda
				\sum_{{i ,\, \Lambda_i \notin \sD{\Lambda}}} n_\Lambda(\Lambda_i)\,\Lambda_i\, p_{(\Lambda_1, \cdots, \overline{\Lambda_i - 1}, \cdots)} = \sum_\Omega z_\Omega^{-1} p_\Omega = \sum_\Omega \zeta_\Omega^{-1} z_\Omega^{-1} p_\Omega  ,
\end{align} 
where the sum is over all $\Omega$'s of degree $n-1$ with exactly one overlined entry and one underlined entry, and where the last equality holds since there is no bilined entry in those $\Omega$'s.

Now, it is immediate that in \eqref{eq104} we have			
			\begin{align}
	z_{(\Lambda_1, \cdots, \sTD{\Lambda_k -1}, \cdots)} = z_\Lambda \qquad \text{and} \qquad
		  		\zeta_{(\Lambda_1, \cdots, \sTD{\Lambda_k -1}, \cdots)}  = \frac{1!}{(\Lambda_k -1 +1)^{1}} = \Lambda_k^{-1},
			\end{align}
which finally gives from \eqref{eq105} that 
			\begin{align}
				\sTD{h}_{n-1} = \sum_{\Lambda \vdash (n-1 | 1,1)} \zeta_\Lambda^{-1} z_\Lambda^{-1} p_\Lambda.
			\end{align}

		\end{proof}
 	\section{Generalisation to arbitrary  $\mathcal{N}$}

	In this final section, {we extend} the previous $\Nm=2$ results to an arbitrary value of  $\mathcal{N}$.
	This requires the introduction of $\mathcal{N}$ types of anticommuting variables:
		\begin{align}
			\theta^{(t)}_j, \;\text{ with } t = 1, \dots, \mathcal{N},\; \text{ and }\; j=1, \dots, N,
		\end{align}
		where the upper index $t$ distinguishes the different types  of variables (replacing {the $\theta$ and $\phi$ variables} in the $\Nm=2$ case). An arbitrary coordinate in this generic superspace is denoted 
		\begin{align}
			(x,\theta) = (x_1, \dots, x_N, \theta^{(1)}_1,\dots,\theta^{(1)}_N,\dots, \theta^{(\mathcal{N})}_1, \dots, \theta^{(\mathcal{N})}_N ).
		\end{align}
		The symmetric superpolynomials {in the ring of polynomials in
these variables} are defined to be those polynomials which are invariant under the simultaneous exchange of all type of variables, i.e., under the interchange 
		\beq(x_i, \theta^{(1)}_i, \cdots, \theta^{(\mathcal{N})}_i) \longleftrightarrow (x_j, \theta^{(1)}_j, \cdots, \theta^{(\mathcal{N})}_j).\eeq

	\subsection{Superpartition}
				The generalization {of the concept of superpartition to} an arbitrary $\mathcal{N}$ is straightforward. {We will} still denote by $\Lambda$ a generic $\mathcal{N}$-superpartition 
		\begin{align}
			\Lambda &= (\Lambda_1, \Lambda_2, \cdots, \Lambda_\ell).
		\end{align}
{We now have} $\sum_{t \geq 0} \binom{\mathcal{N}}{t} = 2^{\mathcal{N}}$ different types of parts in the superpartition. An arbitrary part will be denoted by $	\Lambda_k=a^{\{T\}}$, where 
$\{T\}$ is a subset of 
$\{1,\dots,\mathcal N\}$ and $a$ is a nonnegative integer.
We will use the notation $|\Lambda_k|$ for the size $a$ 
of the part.		
 We will refer to $m$ as the fermionic degree of the part and make it explicit by writing $m(\Lambda_k
	)$. Likewise $T(\Lambda_k)$ denotes the {fermionic content} of the part $\La_k$. 	For example, if $\Lambda_k$ is a part of numerical value $4$ and of fermionic content $\{ 5,2,1 \}$ then
		\begin{align}
			\Lambda_k = 4^{\{ 5,2,1 \}}, \; m(\Lambda_k) = 3, \; T(\Lambda_k) = \{ 5,2,1 \}. 
		\end{align}
			Parts with a given  fermionic content can be grouped together to form  a partition:
		\begin{align}
			\Lambda_j\; \in\; \Lambda^{\{ T \}} \text{ iff } \Lambda_j = |\Lambda_j|^{\{ T \}}.
		\end{align}
		By direct analogy with the $\mathcal{N}=2$ case, we will refer to these partitions as the constituent partitions (these are the analogues of $ \sTD{\La},\,\sT{\La}, \,\sD{\La} $ and $\La^0$).
		 The constituent partitions are naturally separated into 
two families according to the parity of $m$. 
	The parity entails
		 the following restrictions on the parts of the constituent partitions
		\begin{align}
			\forall \; \Lambda_j, \Lambda_k \; \in \; &\Lambda^{\{ T_m \}} \\
			 m \text{ even : }&	\Lambda_j \geq \Lambda_k \geq  0\; \forall\; j>k \\
			 m \text{ odd : }& \Lambda_j > \Lambda_k \geq  0\; \forall\; j>k.
		\end{align}
	where $T_m$ stands for
an arbitrary  fermionic content but of fixed fermionic degree $m$.\\

	These restrictions on constituent partitions directly yield the following generating function of superpartitions: 
	\begin{align}
		\sum_{\substack{n \geq 0 \\ M = (m_1, \cdots, m_\mathcal{N})}}p(n,M) q^n 
		u_{1}^{m_1}\cdots u_{\mathcal{N}}^{m_\mathcal{N}}
		&= \prod_{n=0}^{\infty} \frac{1}{(1-q^{n+1})}
					\dfrac{\prod\limits_{t \in T_{\text{odd}}}(1+u_{\{ t  \}} q^n)}
			{\prod\limits_{t \in T_{\text{even}}}(1-u_{\{ t  \}} q^n)} ,
			\end{align}
	where $T_{\text{even}}$ and $T_{\text{odd}}$ denote the sets $T_m$ for $m$ even and odd respectively. 	Also, given $\{t\} = \{ t_1, t_2, \cdots, t_k \}$, 
$u_{\{ t  \}}$ stands for $ u_{t_1} u_{t_2} \cdots u_{t_k}$. \\

		We now specify the ordering of the parts of a superpartition
we will use. Let $\xi$ and $\zeta$ denote arbitrary parts in some generic superpartition. Then
		\begin{align}
		\begin{array}{l l}
			&\xi \text{ appears before } \zeta \\
			&\quad\text{if } |\xi| > |\zeta|,\\
			&\quad\text{or if } |\xi| = |\zeta| \text{ and } m(\xi) > m(\zeta), \\
			&\quad\text{or if } |\xi| = |\zeta|,\; m(\xi) = m(\zeta) \text{ and } T_i(\xi) > T_i(\zeta) \\
			&\quad\text{ with $i$ being the smallest value for which } T_k(\xi) \neq T_k(\zeta).
		\end{array}
		\end{align}\\

				The bosonic and total fermionic degrees of a partition are respectively, 
		\begin{align}
			|\Lambda| 	&= \sum_{i=1}^{\ell (\Lambda)} |\Lambda_i| \\
			M(\Lambda)	&= \sum_{i=1}^{\ell (\Lambda)} m(\Lambda_i),
		\end{align}
		where, as before, $\ell(\La)$ refers to
		the length of the superpartition $\Lambda$.
		As an example, consider the following superpartition
		\begin{gather}
			\Lambda = ( 5^{\{ 2,1 \}}, 4, 3^{\{ 5,4,1 \}}, 3^{\{ 4,1 \}}, 2^{\{ 6,4,1 \}}, 2^{\{ 6,3,2 \}},2, 1^{\{ 4,3,2,1 \}}, 1^{\{ 4,3,2,1 \}}, 1^{\{ 1 \}}, 0^{\{ 5,2,1 \}}, 0^{\{ 6,1 \}}), \label{eq:Ex.Nspart}\\
			|\Lambda|= 24,\nonumber \\
			M(\Lambda)= 27. \nonumber
		\end{gather}
		This example also illustrates the rule for the ordering of the parts in $\La$.

	\subsection{Symmetric monomial  superfunctions}
	We now introduce the monomial symmetric function $m_\La$ for
generic $\Nm$. Its definition is expressed in terms of the following compact notation:
		\begin{align}
						\theta^{\{ T(\Lambda_i) \}}_i &= \theta^{\{ t_1, t_2, \cdots, t_{m(\Lambda_i)} \}}_i = \theta^{(t_1)}_i \cdots \theta^{{(t_{m(\Lambda_i)})}}_i.
		\end{align}

		\begin{definition} To every $\Lambda \in \text{SPar}_\mathcal{N}$, we associate the monomial symmetric function
		\begin{align}
			m_\Lambda = \sideset{}{'}\sum_{\omega \in S_N}\mathcal{K}_\omega \prod_{i=1}^{\ell(\Lambda)} \theta^{\{ T(\Lambda_i) \}}_i x^{|\Lambda_i|}_i,
		\end{align}
		where $\mathcal{K}_\omega$ permutes the indices of the variables and the prime indicates that the summation is restricted to distinct permutations. 
		\end{definition}
		\begin{example}
			Consider the following superpartition $(5^{\{6,2,1\}} , 3^{\{1\}} , 1^{\{5,3\}} , 1 , 0^{\{6,5\}})$. The  associated monomial is
			\begin{align}
				m_{(5^{\{6,2,1\}} , 3^{\{1\}} , 1^{\{5,3\}} , 1 , 0^{\{6,5\}})} &= \sideset{}{'}\sum_{\omega \in S_N}\mathcal{K}_\omega
				\theta^{\{6,2,1\}}_1\theta^{\{1\}}_2\theta^{\{5,3\}}_3\theta^{\{6,5\}}_5 x^5_1 x^3_2 x_3 x_4\nonumber \\
				&=\sideset{}{'}\sum_{\omega \in S_N}\mathcal{K}_\omega
			\theta^{(6)}_1 \theta^{(2)}_1 \theta^{(1)}_1 \theta^{(1)}_2 \theta^{(5)}_3 \theta^{(3)}_3 \theta^{(6)}_5\theta^{(5)}_5 x^5_1 x^3_2 x_3 x_4
.			\end{align}
		\end{example}

	\subsection{Multiplicative bases}
		Denote as before  a generic multiplicative basis by $f_\Lambda(z,\theta)$, so that
		\begin{align}
			f_\Lambda = \tilde{f}_{\Lambda_1}\cdots\tilde{f}_{\Lambda_\ell}
		\end{align}
		where
		\begin{align}
			\tilde{f}_{\Lambda_i} = f^{\{ T_m \}}_{n},\text{ for } \Lambda_i = n^{\{T_m\}},
		\end{align}
		where $m$ is the fermionic degree of the part (or equivalently the length of $T_m$) and $f^{\{\varnothing\}}_n := f_n$. Each of those components are defined as 
		\begin{align}
			f^{\{ T_m \}}_{n} \propto \left\llbracket d^{ \{ T_m \} }f_{n+m} \right\rrbracket,
		\end{align}
		(extending in  a natural way   to generic $\Nm$ the meaning of the notation $\ll\,\rr$ defined in \eqref{NOTA}) and with
		\begin{align}
			d^{ \{ T_m \} }:= \prod_{t \in T_m} d^{(t)} ,
		\end{align}
where $d^{\{\varnothing\}}=1$.\\

We will use the following normalization for the three multiplicative bases:
		\begin{align}
			p_n^{\{ T_m \}} &:= \frac{n!}{(n+m)!} \left\llbracket d^{ \{ T_m \} } p_{n+m} \right\rrbracket,\\
			h_n^{\{ T_m \}} &:= \left\llbracket d^{ \{ T_m \} } h_{n+m} \right\rrbracket,\\
			e_n^{\{ T_m \}} &:= \left\llbracket d^{ \{ T_m \} } e_{n+m} \right\rrbracket.
		\end{align}
The normalization chosen for the power sums is dictated by the naturalness of the resulting scalar product induced by the reproducing kernel, to which we now turn.

	\subsection{Scalar product and reproducing kernel}
		\begin{definition} 
			For $\Lambda,\Omega\,\in\,\text{SPar}_{\mathcal{N}}$, we define the following scalar product :
			\begin{align}
				\braket{p^\top_\Lambda | p_\omega}= z_\Lambda \delta_{\Lambda,\Omega},
			\end{align}
			where			\begin{align}\label{zLA}
				z_\Lambda &:= \prod_{i=0}^{2^{\Nm}-1} z^{(m(\Lambda^{(i)}))}_{\Lambda^{(i)}}, \\
				z^{(m)}_\lambda &:= \prod_{k \geq 0 } (n_{\lambda}(k))! \left( \frac{k!}{(k+m-1)!} \right)^{n_{\lambda}(k)}, \\
				p^\top_\Lambda &= (-1)^{\binom{M(\Lambda)}{2}}p_\Lambda ,
			\end{align}
and where we recall that the $\Lambda^{(i)}$'s stand for  the constituent partitions and that
$n_{\lambda}(k)$ is the multiplicity of the  part $k$ in the partition $\la$.	
		\end{definition}
		\begin{theorem}
			Let $K(x,y,\theta,\phi)$ be the bi-symmetric formal power series given by
			\begin{align}
				K(x,y,\phi,\theta) = \prod_{ij} \frac{1}{1-x_iy_j - \sum_{t=1}^{\mathcal{N}}\theta_i^{(t)}\phi_j^{(t)}}.
			\end{align}
We then have the following power-sum expansion
			\begin{align}
				K(x,y,\phi,\theta) = \sum_{\Lambda \vdash \text{SPar}_{\mathcal{N}}} z_\Lambda^{-1} (-1)^{\binom{M(\Lambda)}{2}}p_{\Lambda}(x,\theta) p_{\Lambda}(y,\phi).
			\end{align}			
		\end{theorem}
		\begin{proof} Writing $K := K(x,y,\theta,\phi)$, we have
			\begin{align}
			K&=\prod_{ij} \frac{1}{1-x_iy_j - \sum_{t=1}^{\mathcal{N}}\theta_i^{(t)}\phi_j^{(t)}}\nonumber \\
			&= \exp\l(\sum_{i,j} \ln\l[1-x_iy_j -\sum_{t=1}^{\mathcal{N}}\theta_i^{(t)}\phi_j^{(t)}\r]^{-1}\r)\nonumber \\
			&= \exp\l(\sum_{i,j}\sum_{n\geq 1} \frac{1}{n}\l[\, x_i y_j + 
			\sum_{t=1}^{\mathcal{N}} \theta_i^{(t)} \phi_j^{(t)}\r]^n\r)\nonumber \\
			&=\exp\l(\sum_{\substack{i,j \\ n\geq 1}} \frac{(x_iy_j)^n}{{n}} + 
			\sum_{\substack{i,j \\ n\geq 1}} \sum_{m=1}^{n} \binom{n}{m} \frac{(x_iy_j)^{n-m}}{{n}} \l[\sum_{t=1}^{\mathcal{N}}\theta_i^{(t)}\phi_j^{(t)}\r]^m\r) .\label{eq:KernelN_A}
			\end{align}
			Let us concentrate on the exponent of the sum on Grassmannian variables. Using the multinomial theorem, we get
			\begin{align}
				\left[\sum_{t=1}^{\mathcal{N}} \theta_i^{(t)}\phi_j^{(t)} \right]^m &=
					\sum_{\substack{k_1+k_2+\cdots+k_{\mathcal{N}}=m\\ 0\leq k_l \leq m}} \frac{m!}{k_1!k_2!\cdots k_{\mathcal{N}}!} 
					\prod_{ t=1}^{\mathcal{N}}(\theta_i^{(t)}\phi_j^{(t)})^{k_t}\nonumber \\
					&=m! \sum_{\substack{\epsilon_1 +\epsilon_2 + \cdots \epsilon_{\mathcal{N}}=m\\ 0 \leq \epsilon_l \leq 1}} 					\prod_{ t=1}^{\mathcal{N}}(\theta_i^{(t)}\phi_j^{(t)})^{\epsilon_t}.
								\end{align}
			In the second expression, we have changed $k_t$ to $\ep_t$ to stress the fact that it can only take  the values 0 and 1, which implies that $\ep_j!=1$ for all $j$. Bringing this result back into \eqref{eq:KernelN_A}, 	we see that in the exponent, the term that multiplies $\sum_{\ep_j}\prod_t (\theta_i^{(t)}\phi_j^{(t)})^{\epsilon_t}$  can be manipulated as follows,
			\begin{align}
				& \sum_{\substack{i,j\\ n\geq 1 \\ m=1,\ldots,\min(n,\mathcal N)}} \frac{n! m!}{n (n-m)! m!}(x_iy_j)^{n-m}
				= \sum_{\substack{i,j\\ n\geq 1\\ m=1,\ldots,\min(n,\mathcal N)}} \frac{(n-1)!}{(n-m)!} (x_iy_j)^{n-m}
		= \sum_{\substack{i,j\\ n\geq 0 \\ m=1,\ldots,\mathcal N}} \frac{(n+m-1)!}{(n)!} (x_iy_j)^{n}
				.			\end{align}
			Reorganizing the variables we get 
			\begin{align}
				&\sum_{\substack{i,j\\ n\geq 0 \\ m=1,\ldots,\mathcal N}} \frac{(n+m-1)!}{(n)!} (-1)^{\binom{m}{2}} \sum_{\mathcal{N} \geq t_1 > \cdots > t_m \geq 1 } \l(\theta_i^{(t_1)} \theta_i^{(t_2)}\cdots\theta_i^{(t_m)}x_i^n\r) \l(\phi_j^{(t_1)} \phi_j^{(t_2)}\cdots\phi_j^{(t_m)}{y_j^n} \r). 							\end{align}
The above expression unveils the power-sum constituents with the terms in the last two parentheses. 			Now we can introduce the power-sum notation and incorporate the value $m=0$ which refers to the case where there are no Grassmannian variables. Recall that $\{ T_m\}$ stands for the multi-index 
$T_m = \{ t_1, t_2, \cdots, t_m \}$ with $\mathcal{N} \geq t_1 > \cdots > t_m \geq 1 $. The summation over $T_m$ with $m$ fixed followed by a sum over $m$ amounts  to summing over every possible combination of the $t_j$'s (and note that  $T_0 = \{ \varnothing \}$ and $p^{\{ \varnothing\}}_n$ is just the usual power sum). Getting the exponential back and developing the Taylor series we have, 
			\begin{align}
				&\prod_{\substack{i,j\\ n\geq 0 \\ m=0,\ldots,\mathcal{N}}} \exp\left\{  \sum_{T_m} \frac{(n+m-1)!}{(n)!} (-1)^{\binom{m}{2}}  p_n^{\{ T_m \}}(x,\theta) p_n^{\{ T_m \}}(y,\phi) \right\} \nonumber \\
				&=\prod_{\substack{i,j\\ n\geq 0 \\ m=0,\ldots,\mathcal{N} \\ T_m}} \sum_{k} \frac{1}{k!} \left(\frac{(n+m-1)!}{n!} \right)^k \left( (-1)^{\binom{m}{2}} \right)^k  \left( p_n^{\{ T_m \}}(x,\theta) p_n^{\{ T_m \}}(y,\phi) \right)^k.
			\end{align}
			
		From now on, to alleviate slightly the notation, we will denote the constituent partition of fermionic content  $T_m$ as $\La^{T_m}$ instead of $\La^{\{T_m\}}$.\\

	We now focus on a single set $T_m$ and develop the product over $n$ in order to reconstruct the product of power sums for the arbitrary constituent partition $\Lambda^{T_m}$. Leaving out a multiplying  factor for the moment, let us consider the expression
			\begin{align}
				(-1)^{\ell(\Lambda^{T_m})\binom{m}{2}} p_{\Lambda^{T_m}_1}^{\{ T_m \}}(x,\theta) p_{\Lambda^{T_m}_1}^{\{ T_m \}}(y,\phi)p_{\Lambda^{T_m}_2}^{\{ T_m \}}(x,\theta) p_{\Lambda^{T_m}_2}^{\{ T_m \}}(y,\phi)\cdots p_{\Lambda^{T_m}_{\ell(\Lambda^{T_m})}}^{\{ T_m \}}(x,\theta) p_{\Lambda^{T_m}_{\ell(\Lambda^{T_m})}}^{\{ T_m \}}(y,\phi).
			\end{align}
			Reorganizing the product of power sums to move the $(y,\phi)$ variables to the right gives the familiar sign $(-1)^{\binom{\ell(\Lambda)}{2}}$ if $m$ is odd and no sign otherwise. Both parities are thus taken into account by writing the sign as $(-1)^{m\binom{\ell(\Lambda)}{2}}$. The above expression becomes then
			\begin{align}
				(-1)^{\ell(\Lambda^{T_m})\binom{m}{2}} (-1)^{m\binom{\ell(\Lambda^{T_m})}{2}}p_{\Lambda^{T_m}}^{\{ T_m \}}(x,\theta) p_{\Lambda^{T_m}}^{\{ T_m \}}(y,\phi).
			\end{align} 
Before we pursue, we modify once again our notation to make it a little more handy. Here the $T^{(j)}$'s stand for the different fermionic contents irrespectively of their fermionic degree. Therefore, the index $j$ runs from 0 to  $2^{\Nm}-1$.
			Now, we know that the fermionic degree of a superpartition is the sum of the fermionic degrees of its parts. So in the case where we are dealing with a constituent partition $\Lambda^{T^{(j)}}$ we know that all parts share the same fermionic content and hence are of the same fermionic degree. This means that the fermionic degree of the partition $\Lambda^{T^{(j)}}$ is just its length times the fermionic degree ($m_j$) of its parts. 
As a further step towards alleviating  the notation, let us denote $\Lambda^{T^{(j)}}$ as $\Lambda^{(j)}$, which amounts to express a generic superpartition as $\Lambda = (\Lambda^{(2^{\mathcal{N}-1})}; \cdots;\Lambda^{(j)}; \cdots; \Lambda^{(0)})$. Therefore, writing $\ell_j$ for $\ell(\Lambda^{(j)})$ for short, the fermionic degree of the constituent partition $\Lambda^{(j)}$ is: 
			\begin{align}
				M_j := \ell_jm_j.
			\end{align}
			Given this, we can craft a better expression for the exponent of $-1$.    We note that
			\begin{align}\frac12\ell_j m_j(m_j-1) +\frac12m_j\ell_j(\ell_j-1)&=\frac12M_j(M_j-1)\quad \mod 2.
			\end{align}
		{	Indeed, the difference between the terms on the  two sides of the previous expression is
			\beq \frac12 \ell_j(\ell_j-1)m_j(m_j-1)=0 \quad \mod 2\eeq
			since the product of two consecutive integers is necessarily even.}\\

			We are now in a position to reintroduce the full expression of the kernel
			\begin{align}
				K &= \sum_{\Lambda} z_\Lambda^{-1} (-1)^{\sum_{i =0}^{2^{\mathcal{N}}-1} \binom{M_i }{2}}
				 \prod_{i=0}^{2^{\mathcal{N}}-1}
				p_{\Lambda^{(i)}}^{\{ T^{(i)} \}}(x,\theta) p_{\Lambda^{(i)}}^{\{ T^{(i)} \} }(y,\phi),
			\end{align}
			where $z_\La$ is defined in \eqref{zLA}.
								A final step of reorganization consists  in bringing						all the power sums in the variables $(y,\phi)$ to the right. This introduces a further sign: $(-1)^{\sum_{i>j}M_iM_j}$. Combining the two sign factors, we get			\begin{align}
			\frac12\sum_i M_i(M_i-1)+ \sum_{i>j}M_iM_j &=\frac12\sum_i M_i(M_i-1) + \frac{1}{2}\l(\sum_{i,j}M_iM_j- \sum_i M_i^2\r)= \frac{M(M-1)}{2},
			\end{align}
		where  \beq M= \sum_i M_i\eeq
is the total fermionic degree of the superpartition $\Lambda$.
			Hence, we finally obtain 
			\begin{align}
				K &= \sum_{\Lambda \vdash \text{SPar}_{\mathcal{N}}} z_\Lambda^{-1} (-1)^{\binom{M(\Lambda)}{2}}
				p_{\Lambda}(x,\theta) p_{\Lambda}(y,\phi)	.	\end{align}
		\end{proof}

		We conclude this section with the presentation of three propositions whose proofs are omitted since they parallel closely their $\mathcal{N}=2$ analogues.
		\begin{proposition} 
			$K(x,y,\theta,\phi)$ is a reproducing kernel in the space of symmetric superfunctions: 
			\begin{align}
				\braket{K^\top(x,y,\theta,\phi) | f(x,\theta)} = f(y,\phi),\; \forall \; f \in \mathscr{P}^{S_\infty} 
			\end{align}
			where $\top$ acts only on the variables $\theta$. 
		\end{proposition}
		\begin{proposition}
			Let $u_\Lambda$ and $v_\Lambda$ be two bases of $\mathcal{P}^{S_\infty}_{(n|M)}$ (using a self-explanatory notation).  Then we have
			\begin{align}
				K(x,y,\theta,\phi) = \sum_{\text{SPar}}u^{\top}_\Lambda(x,\theta) v_\Lambda(y,\phi) 
				\quad \Longleftrightarrow \quad
				\braket{u_\Lambda^\top | v_\Omega} = \delta_{\Omega \Lambda}.
			\end{align}
		\end{proposition}
		\begin{proposition} \label{mhdualNarbit}
			The monomial and the complete symmetric functions are dual, that is
			\begin{align}
				K(x,y,\theta,\phi) = \sum_{\text{SPar}} m^{\top}_\Lambda(x,\theta) h_\Lambda(y,\phi),
			\end{align}
			or equivalently 
			\begin{align} 
				\braket{m_\Lambda^\top | h_\Omega} = \delta_{\Omega \Lambda} .
			\end{align}
		\end{proposition}
				
 	\section{Concluding remarks}

\subsection{A further line of generalization: deformation of the scalar product}
It is interesting to point out that the scalar product \eqref{PS1} can be deformed by the introduction of a free parameter $\a$ as follows:
			\begin{align}\label{alphaPS}
				\braket{p_\Lambda^\top | p_\Omega}_\a = \alpha^{\ell (\Lambda)} z_\Lambda \zeta_\Lambda \delta_{\Lambda \Omega},
			\end{align}
where $z_\La$ and $\zeta_\La$ are defined in \eqref{zdef} and \eqref{zetadef} respectively. This particular form is actually motivated by the following expression of the  reproducing kernel
			\begin{align}
K(x,y,\phi,\psi,\theta,\tau;\alpha)&=\prod_{i,j}\frac{1}{(1-x_iy_j -\phi_i\o\tau_j -\theta_i\u\tau_j)^{1/\alpha}}
			\end{align}
			which can be decomposed in the form
			\begin{align}
	K(x,y,\phi,\o\tau,\theta,\u\tau;\alpha)				&= \sum_{ \Lambda \,\in\, \text{SPar} } \alpha^{-\ell(\Lambda)} (-1)^{\binom{\sT{m}+\sD{m}}{2}} z_\Lambda^{-1} \zeta_\Lambda^{-1} p_\Lambda(x,\phi,\theta)p_\Lambda(y,\o\tau,\u\tau).\label{eq:alphaKernel}
			\end{align}
		  	The correctness of this expression for the kernel is justified {\it a priori} by the naturalness of the extension from 
$\Nm=1$ to $\Nm=2$ and {\it a fortiori} by the duality between the bases $m_\La$ and $g^{(\a)}_\Lambda$, where the latter is
the $\a$-deformation of $h_\Om$ which is also the ${\mathcal N}=2$ extension of $g^{(\a)}_\lambda$ defined as \cite{Macdonald1998}
		\begin{align}
			g^{(\a)}_\lambda = \prod_i g_{\lambda_i}^{(\a)}(x)
		\qquad\text{with}\qquad G(x,t;\alpha)  = \sum_{n \geq 0 } g_n^{(\a)} t^n = \prod_i (1-x_i t)^{-1/ \alpha}. \label{eq:GenGalphalambda}
		\end{align}
In the ${\mathcal N}=2$ case, the generating function becomes
			\begin{align}
			\ti G(x,\phi,\theta,t,\o\tau,\u\tau;\alpha) =\ll \O G(x,t;\a)\rr
				&=	\prod_i \frac{1}{(1 -tx_i - \o\tau\phi_i-\u\tau\theta_i )^{1/\alpha}} \nonumber \\
				&= \sum_{n\geq 0} t^n(g_n^{(\a)} +\o\tau\, \sT{g}^{(\a)}_{n} + \u\tau\, \sD{g}^{(\a)}_{n}- \o\tau\,\u\tau  \,\sTD{g}^{(\a)}_{n}). \label{eq:GLambdaAlphaGen}
			\end{align}		
		On the other hand, we have
		\beq 
			K(x,y,\phi,\o\tau,\theta,\u\tau;\alpha)= 
			\prod_i \left( 
			\sum_{n\geq 0} x_i^n(g_n^{(\alpha)} {+} \phi_i \sT{g}^{(\alpha)}_{n} {+} \theta_i \sD{g}^\alpha_{n} - \phi_i\theta_i  \sTD{g}^{(\alpha)}_{n})
			\right).
			\eeq
Following the proof of Proposition~\ref{mhdual}, we have immediately that
\begin{align} \nonumber \\ 
\phantom{\sum}K(x,y,\phi,\o\tau,\theta,\u\tau;\alpha)  
		&= \sum_\Lambda (-1)^{\binom{\sT{m}+\sD{m}}{2}} m_\sL(x,\phi,\theta) g^{(\alpha)}_\sL(y, \o\tau, \u\tau)
			\end{align}
which implies that the basis $g^{(\alpha)}_\Lambda$ is dual to the monomial basis:
			\begin{align}
				\braket{m_\Om^\top|g^{(\alpha)}_\Lambda }_\a = \delta_{\Lambda, \Omega}.
			\end{align}

One can go a step further and introduce the $\a$-deformation of the homomorphism \eqref{eq:involution} 
defined by
			\begin{align}
				\omegahat_\alpha : (p_n,\sT{p}_n ,\sD{p}_n ,\sTD{p}_n ) &\longrightarrow (-1)^{n-1}\alpha\,(p_n,-\sT{p}_n ,-\sD{p}_n ,\sTD{p}_n )
															\end{align}
resulting into
			\begin{align}
				\omegahat_\alpha(p_\Lambda) = \omega^\alpha_\Lambda\, p_\Lambda, \quad \omega^\alpha_\Lambda \equiv (-\alpha)^{|\Lambda| - \ell(\sTDL) - \ell(\Lambda^0)}.
			\end{align}
			It is then a simple matter to verify that \beq \label{dual}\omegahat_\alpha( g^{(\a)}_\La)=e_\La,\eeq
			which is a further confirmation that this new basis deforms naturally the homogeneous one. 

\subsection{Toward the $\Nm=2$ Jack superpolynomials}
As mentioned in the introduction, the end objective of this work is the construction of the $\Nm=2$ Jack superpolynomials.
These would be tentatively defined as in the  $\Nm=0,1$  cases, via the following two conditions: triangularity in the monomial basis and orthogonality.  The orthogonality condition would be defined with respect to the scalar product \eqref{alphaPS}. Moreover, the three partners of  $g^{(\a)}_n$ are candidates for the particular ${\mathcal N}=2$ Jack superpolynomials that are indexed by a  superpartition represented by a single row diagram:
\begin{equation} P_{\sT{n}}^{(\a)}\propto \sT{g}^{(\alpha)}_n,\qquad P_{\sD{n}}^{(\a)}\propto \sD{g}^{(\alpha)}_n,\quad\text{and}\quad P_{\sTD{n}}^{(\a)}\propto \sTD{g}^{(\alpha)}_n.
\end{equation}
The non-trivial duality relation \eqref{dual} is another piece of evidence of an  underlying Jack-like structure.\\

Unfortunately, at least for any ordering of the superpartitions that we
deem natural, the direct construction of generic $\Nm=2$ Jack superpolynomials failed. 
The difficulty could be related to
a technical missing aspect in the $\Nm=2$ formalism: the definition of a conjugate superpartition. Indeed, the diagrammatical representation of the superpartitions does not entail  an immediate definition of the conjugation operation which, in the $\Nm=0,1$ case, amounts to {interchanging} rows and columns. In the present case, there is a built in row-column asymmetry in that at most one vh-circle is allowed in a row but an arbitrary number of {them} can be piled up in  a column. On the other hand, it could be that some restrictions have to be imposed on the superpartitions and/or
some symmetry requirements need to be enforced on the bases.
Exploratory computations show that by restricting the anticommuting variables with the extra requirement $\ta_i\phi_i=0$ (no sum over $i$) provides a well-defined  extension of Jack polynomials which we intend to study.
\\

\subsection*{Acknowledgements}

 This work was  supported by the Natural Sciences and Engineering Research Council of Canada and the
Fondo Nacional de Desarrollo Cient\'{\i}fico y
Tecnol\'ogico de Chile grant \#1130696.

 	\appendix
	\section{Action of the master operator for generic $\Nm$}\label{apendix:ProofIdent}
 We present the {details} of the proof of the identity \eqref{IGF} but in a generic version applicable to arbitrary $\Nm$.\\

	 We first introduce $\Nm$ distinct exterior derivatives $d^{(q)}$ and generalize the differential operators $\o\D$ and $\u\D$ as follows 
	 \begin{align}
	 	\mathcal{D}^{(q)}f := d^{(q)}t\wedge d^{(q)}f,\qquad q=1,\cdots, \Nm. 
	 \end{align}
	 The master operator becomes
	 \beq \O_{\Nm}:= \prod_{q=1}^{\mathcal{N}}(1+\mathcal{D}^{(q)}).\eeq
	 We also extend the meaning of the $\ll\; \rr$ operation in a natural way:
	 \begin{align}
	 	 &\left\ll F(x_i,d^{(1)} x_i,\cdots,d^{(\mathcal{N})}x_i,d^{(1)} t,\cdots,d^{(\mathcal{N})}t) \right\rr:=\nonumber \\
	 	 &\l[F(x_i,d^{(1)} x_i,\cdots,d^{(\mathcal{N})}x_i,d^{(1)} t,\cdots,d^{(\mathcal{N})}t)\r]_{\substack{d^{(p)} x_i\to \theta^{(p)}_i \\ t d^{(p)}t\to \tau^{(p)}}},
	 \end{align}
where $\tau^{(p)}$ are 	distinct anticommuting but constant parameters. The general version of the identity we want to prove is
	 	 \beq \ll\O_{\Nm} F(xt)\rr =F\l(xt+\textstyle{\sum_p\tau^{(p)}\theta^{( p )}}\r). \label{GenI}
		 \eeq\\

	Consider first the  series expansion of $F\l(xt+\sum_p\tau^{(p)}\theta^{( p )}\r)$ around the points $\xi_i=x_it+\sum_{p\ne q}\tau^{(p)}\theta_i^{(p)}$:
		 \begin{align}
	 	F\l(tx+\textstyle{\sum_p\tau^{(p)}\theta^{(p)}}\r) 	 	= 	F\l(tx+\textstyle{\sum_{p \neq q} \tau^{(p)}\theta^{(p)}}\r) 
	 		+\sum_{i=1}^{N} \tau^{(q)}\theta^{(q)}_i \partial_{\xi_i}
			F\l(tx+\textstyle{\sum_{p \neq q} \tau^{(p)}\theta^{(p)} }\r). 	\label{eq:GenNtaylorSeries}
	 \end{align}
	 Note that the series truncates after the second term since $(\tau^{(q)})^2=0$.
	 \\

	 On the other hand, for the left-hand side of \eqref{GenI}, we proceed recursively.
Suppose that we have:
	 \beq \l\ll\prod_{p\ne  q}(1+\mathcal{D}^{( p )})F(xt)\r\rr=F\l(tx+\textstyle{\sum_{p \ne q} \tau^{(p)}\theta^{(p)} }\r).\eeq
	We then evaluate the action of the operator $(1+\mathcal{D}^{(q)})$ on this expression:
		 \begin{align}
	 	& (1+\mathcal{D}^{(q)}) F\l(xt+\textstyle{\sum_{p \ne q}\tau^{(p)}\theta^{( p )}}\r)  \nonumber \\
	 	 &=  F\l(xt+\textstyle{\sum_{p \ne q}\tau^{(p)}\theta^{( p )}}\r)  + (d^{(q)}t\wedge d^{(q)})F\l(xt+\textstyle{\sum_{p \ne q}\tau^{(p)}\theta^{( p )}}\r)  \nonumber \\
	 	 &=F\l(xt+\textstyle{\sum_{p \ne q}\tau^{(p)}\theta^{( p )}}\r)  + \sum_{i=1}^{N} (d^{(q)}t\wedge d^{(q)}x_i)\partial_i F\l(xt+\textstyle{\sum_{p \ne q}\tau^{(p)}\theta^{( p )}}\r) ,	 		 \end{align}
so that
\begin{align}  \l\ll\prod_{p}(1+\mathcal{D}^{( p )})F(xt)\r\rr&=F\l(xt+\textstyle{\sum_{p \ne q}\tau^{(p)}\theta^{( p )}}\r)  + \sum_{i=1}^{N} \tau^{(q)}\ta_i^{(q)}\partial_{\xi_i} F\l(xt+\textstyle{\sum_{p \ne q}\tau^{(p)}\theta^{( p )}}\r) \nonumber \\&=F\l(tx+\textstyle{\sum_p\tau^{(p)}\theta^{(p)}}\r) ,
 \end{align}
 where we used eq. \eqref{eq:GenNtaylorSeries}
 in the last step. Letting now $p$ {run} from 1 to $\Nm$ completes the proof of \eqref{GenI}.

 	\section{Monomial Product Algorithm} \label{sect:monomialProd}
	We present here a simple algorithm to compute the monomial-basis expansion of the product of two {monomial} superfunctions in terms of monomial superfunctions. 
	In particular, this algorithm can be used to provide an efficient
method to convert power sums into monomials when using a symbolic computation software. {In the present context, it has an immediate application, being the starting point of the proof that our multiplicative ``bases'' are genuine bases, which is the subject of the following appendix.}
In order to introduce the algorithm, we first have to introduce the operation of addition of superpartitions.

	\subsection{Addition of superpartitions}
			The addition of the two superpartitions $\Lambda$ and $\Omega$ is
	\begin{align}
		\Lambda + \Omega = (\Lambda_1 + \Omega_1, \Lambda_2 + \Omega_2, \cdots, \Lambda_\ell + \Omega_\ell), \quad \text{where}\quad \ell=\text{max}\,\bigl(\ell(\La),\ell(\Om)\bigr).
	\end{align}
																																																							If the lengths of the superpartitions are not equal, we add the appropriate number of zeros to the shortest superpartition. Now, we must specify  the rule for partwise addition. The rule is most easily understood from the diagrammatic representation. When adding two parts, 
we simply concatenate the symbols representing the parts  and reorder them according to our diagrammatical conventions. The various cases are illustrated as follows:
	\begin{alignat}{3}\label{addition}
	&\superY{\,&\none[{\scriptstyle \cdots}]&\,} &+ &\superY{\,&\none[{\scriptstyle \cdots}]&\,} 	&= &\superY{\,&\none[{\scriptstyle \cdots}]&\,&\,&\none[{\scriptstyle \cdots}]&\,}\nonumber\\
	&\superY{\,&\none[{\scriptstyle \cdots}]&\,&\yT} &+ &\superY{\,&\none[{\scriptstyle \cdots}]&\,} 	&= &\superY{\,&\none[{\scriptstyle \cdots}]&\,&\,&\none[{\scriptstyle \cdots}]&\,&\yT}\nonumber\\
	&\superY{\,&\none[{\scriptstyle \cdots}]&\,&\yC} &+ &\superY{\,&\none[{\scriptstyle \cdots}]&\,} 	&= &\superY{\,&\none[{\scriptstyle \cdots}]&\,&\,&\none[{\scriptstyle \cdots}]&\,&\yC}\nonumber\\
	&\superY{\,&\none[{\scriptstyle \cdots}]&\,&\yC} &+ &\superY{\,&\none[{\scriptstyle \cdots}]&\,&\yT} 	&= &\superY{\,&\none[{\scriptstyle \cdots}]&\,&\,&\none[{\scriptstyle \cdots}]&\,&\yT&\yC} \nonumber \nonumber\\
																	&&&&= &\superY{\,&\none[{\scriptstyle \cdots}]&\,&\,&\none[{\scriptstyle \cdots}]&\,&\yTC}\nonumber\\
	&\superY{\,&\none[{\scriptstyle \cdots}]&\,&\yTC} &+ &\superY{\,&\none[{\scriptstyle \cdots}]&\,} 	&= &\superY{\,&\none[{\scriptstyle \cdots}]&\,&\,&\none[{\scriptstyle \cdots}]&\,&\yTC}.
	\end{alignat}
	Whenever a part obtained by juxtaposing the symbols does not meet the restrictions reflecting the constraints on superpartitions, the diagram is annihilated. For example, two circles  of the same type cannot be juxtaposed. 

	\begin{example}Let $\Lambda = (\sT{3},\sTD{2},\sD{2},1,1,\sT{0})$ and $\Omega=(\sD{1},1,1,\sT{0})=(\sD{1},1,1,\sT{0},0,0)$ and let the zeros be represented on the diagram as $\emptyset$, then
	\begin{align}
	\Lambda + \Omega = 
	\superY{\,& \,& \,& \yT\\
    \,& \,& \yTC\\
    \,& \,& \yC\\
    \,\\
    \,\\
    \yT}
    + 
    \superY{\,& \yC\\
    \,\\
    \,\\
    \yT\\
    \none[{\vphantom{A}_\emptyset}]\\
    \none[{\vphantom{A}_\emptyset}]}
    = 
    \, \superY{\,& \,& \,& \,& \yTC\\
    \,& \,& \,& \yTC\\
    \,& \,& \,& \yC\\
    \,& \yT\\
    \,\\
    \yT}.
	\end{align}
	As another example let $\Lambda= (\sT{1},\sD{0})$ and $\Omega=(1,\sD{0})$ then
	\begin{align}
		\superY{\,&\yT\\ \yC} + \superY{\,\\ \yC} = 0,
	\end{align}
	since $\superY{\yC &\yC}$ is not allowed.  
	\end{example}

	\subsection{The algorithm}
		Given two monomial symmetric functions $m_\Lambda$ and $m_\Omega$, one has 
		\begin{align}
			m_\Lambda m_\Omega = \sum_{\Gamma \in \mathcal{P}} \epsilon_{\Lambda \Omega}^\Gamma N_{\Lambda \Omega}^\Gamma m_\Gamma. 
		\end{align}
		where $N_{\Lambda \Omega}^\Gamma$ is a certain positive number, $\epsilon_{\Lambda \Omega}^\Gamma$ is either $+1$ or $-1$ and $\mathcal{P}$ is the set composed of every superpartition corresponding to a  filled diagram obtained with the algorithm that follows. \\

		Given two arbitrary superpartitions $\Lambda$ and $\Omega$ of length $\ell(\Lambda)$ and $\ell(\Omega)$ respectively, the resulting sum of monomials is given by the following procedure : 
		\begin{enumerate}
			\item We add zeros at the end of each superpartition
in such a way that they each have length $\ell(\Lambda)+\ell(\Omega)$, that is $\Lambda \mapsto \tilde \La=(\Lambda, 0^{\ell(\Omega)})$ and $\Omega \mapsto \tilde \Om=(\Omega, 0^{\ell(\Lambda)})$.
			\item The boxes of the diagram $\ti\Lambda$ are filled with the letter $a$ and its circles are transformed as follow : 			\begin{align}
			\superYbig{\yT} 	& \longrightarrow 	\superYbig{\yF{ \sT{a_i} }}\\
			\superYbig{\yC} 	& \longrightarrow 	\superYbig{\yF{ \sD{a_i} }}\\
			\superYbig{\yTC}	& \longrightarrow 	\superYbig{\yF{ \sT{a} }& \yF{ \sD{a} } } 
			\end{align} 
			where the subscript  starts at 1 for the highest circle in the diagram and increases downward.
			\item The previous step is also applied to the  $\ti\Omega$ diagram with $a$ replaced by $b$.
			\item We then permute the rows (including the $\emptyset$ rows) of the  filled $\ti\Omega$ diagram in every distinct way. The resulting diagrams are then added to the $a$-filled $\ti\Lambda$ diagram. For each resulting diagram, we rearrange the rows so that every diagram becomes that of a {\it bona fide} filled superpartition. The boxes containing $a$'s are to be placed to the left of those containing $b$'s. When two types of circles are next to each other, the one with an overlined letter is placed to the left of the one marked by an underlined letter, irrespectively of the $a$ vs $b$ order. 
			\item In the resulting set of filled diagrams, we keep only distinct ones. 
			\item For a given filled diagram in that set, associated to let's say  a superpartition $\Gamma$, the number $N_{\Lambda \Omega}^\Gamma$ is the number of distinguishable (filling-wise) ways of permuting rows so that the superpartition ordering on parts is still respected.
									\item The sign $\epsilon_{\Lambda \Omega}^\Gamma$ is given by the following procedure: reading the diagram $\Gamma$ from top to bottom and left to right, we write the corresponding sequence of  $a_i$'s and $b_j$'s. We then  determine the minimal number of elementary permutations needed to reorder that sequence so that all $a$'s appear first with increasing subscripts, followed by the $b$'s, also with increasing subscripts. Denote this number by $r$; then $\epsilon_{\Lambda \Omega}^\Gamma = (-1)^r$. 
			\item Repeating the last three steps for each filled diagrams yields the complete monomial decomposition. 
		\end{enumerate}

		\subsection{An illustrative example}

		We clarify the procedure with the following example: find the monomial expansion of the product $
				m_{\, \superYsmall{\,& \yT\\
				\,& \yC\\
				\yC}}\,
				m_{\, \superYsmall{\,& \yT\\
				\yT}}
			$.
			Taking $\Lambda = \superYsmall{\,& \yT\\ \,& \yC\\ \yC}$ and $\Omega = \superYsmall{\,& \yT\\ \yT}$, we apply the algorithm step by step.
			\begin{enumerate}
			\item Adding zeros
				\begin{align}
					&\superY{\,& \yT\\
					\,& \yC\\
					\yC} 
					\longrightarrow 
					\superY{\,& \yT\\
					\,& \yC\\
					\yC\\
					\none[{\vphantom{A}_\emptyset}]\\
					\none[{\vphantom{A}_\emptyset}]}
					&
					&\superY{\,& \yT\\
					\yT}
					\longrightarrow 
					\superY{\,& \yT\\
					\yT \\
					\none[{\vphantom{A}_\emptyset}]\\
					\none[{\vphantom{A}_\emptyset}]\\
					\none[{\vphantom{A}_\emptyset}]}
				\end{align}
			\item[2.,3.] $a$ filling and $b$ filling
				\begin{align}
					&\superY{\,& \yT\\
					\,& \yC\\
					\yC\\
					\none[{\vphantom{A}_\emptyset}]\\
					\none[{\vphantom{A}_\emptyset}]}
					\longrightarrow
					\superY{a& \yF{\sT{a_1}}\\
					a & \yF{\sD{a_2}} \\
					\yF{\sD{a_3}} \\
					\none[{\vphantom{A}_\emptyset}]\\
					\none[{\vphantom{A}_\emptyset}]}
					&
					&\superY{\,& \yT\\
					\yT \\
					\none[{\vphantom{A}_\emptyset}]\\
					\none[{\vphantom{A}_\emptyset}]\\
					\none[{\vphantom{A}_\emptyset}]}
					\longrightarrow
					\superY{b & \yF{\sT{b_{\scriptscriptstyle 1} }}\\
					\yF{\sD{b_2}} \\
					\none[{\vphantom{A}_\emptyset}]\\
					\none[{\vphantom{A}_\emptyset}]\\
					\none[{\vphantom{A}_\emptyset}]}
				\end{align}
			\item[4., 5.] Permutations and addition
				\begin{align}
						\superY{b & \yF{\sT{b_{\scriptscriptstyle 1} }}\\
						\yF{\sD{b_2}} \\
						\none[{\vphantom{A}_\emptyset}]\\
						\none[{\vphantom{A}_\emptyset}]\\
						\none[{\vphantom{A}_\emptyset}]} \xrightarrow[{\text{Permutations}}]{}
												&\left\{ 
						\superY{b & \yF{\sT{b_{\scriptscriptstyle 1} }}\\
						\yF{\sD{b_2}} \\
						\none[{\vphantom{A}_\emptyset}]\\
						\none[{\vphantom{A}_\emptyset}]\\
						\none[{\vphantom{A}_\emptyset}]},
												\superY{b & \yF{\sT{b_{\scriptscriptstyle 1} }}\\
						\none[{\vphantom{A}_\emptyset}]\\
						\yF{\sD{b_2}} \\
						\none[{\vphantom{A}_\emptyset}]\\
						\none[{\vphantom{A}_\emptyset}]},
												\superY{b & \yF{\sT{b_{\scriptscriptstyle 1} }}\\
						\none[{\vphantom{A}_\emptyset}]\\
						\none[{\vphantom{A}_\emptyset}]\\
						\yF{\sD{b_2}} \\
						\none[{\vphantom{A}_\emptyset}]},
												\superY{b & \yF{\sT{b_{\scriptscriptstyle 1} }}\\
						\none[{\vphantom{A}_\emptyset}]\\
						\none[{\vphantom{A}_\emptyset}]\\
						\none[{\vphantom{A}_\emptyset}]\\
						\yF{\sD{b_2}}},
												\superY{
						\none[{\vphantom{A}_\emptyset}]\\
						b & \yF{\sT{b_{\scriptscriptstyle 1} }}\\
						\yF{\sD{b_2}} \\
						\none[{\vphantom{A}_\emptyset}]\\
						\none[{\vphantom{A}_\emptyset}]},
												\superY{
						\none[{\vphantom{A}_\emptyset}]\\
						b & \yF{\sT{b_{\scriptscriptstyle 1} }}\\
						\none[{\vphantom{A}_\emptyset}]\\
						\yF{\sD{b_2}} \\
						\none[{\vphantom{A}_\emptyset}]},
												\superY{
						\none[{\vphantom{A}_\emptyset}]\\
						b & \yF{\sT{b_{\scriptscriptstyle 1} }}\\
						\none[{\vphantom{A}_\emptyset}]\\
						\none[{\vphantom{A}_\emptyset}] \\
						\yF{\sD{b_2}} },
												\superY{
						\none[{\vphantom{A}_\emptyset}]\\
						\none[{\vphantom{A}_\emptyset}]\\
						b & \yF{\sT{b_{\scriptscriptstyle 1} }}\\
						\yF{\sD{b_2}} \\
						\none[{\vphantom{A}_\emptyset}]}, \right. \nonumber\\
																								& \phantom{\{}\left.
						\superY{
						\none[{\vphantom{A}_\emptyset}]\\
						\none[{\vphantom{A}_\emptyset}]\\
						b & \yF{\sT{b_{\scriptscriptstyle 1} }}\\
						\none[{\vphantom{A}_\emptyset}] \\
						\yF{\sD{b_2}} },
												\superY{
						\none[{\vphantom{A}_\emptyset}]\\
						\none[{\vphantom{A}_\emptyset}]\\
						\none[{\vphantom{A}_\emptyset}] \\
						b & \yF{\sT{b_{\scriptscriptstyle 1} }}\\
						\yF{\sD{b_2}} },
												\superY{
						\yF{\sD{b_2}} \\
						b & \yF{\sT{b_{\scriptscriptstyle 1} }}\\
						\none[{\vphantom{A}_\emptyset}]\\
						\none[{\vphantom{A}_\emptyset}]\\
						\none[{\vphantom{A}_\emptyset}]},
												\superY{
						\yF{\sD{b_2}} \\
						\none[{\vphantom{A}_\emptyset}]\\
						b & \yF{\sT{b_{\scriptscriptstyle 1} }}\\
						\none[{\vphantom{A}_\emptyset}]\\
						\none[{\vphantom{A}_\emptyset}]},
												\superY{
						\yF{\sD{b_2}} \\
						\none[{\vphantom{A}_\emptyset}]\\
						\none[{\vphantom{A}_\emptyset}]\\
						b & \yF{\sT{b_{\scriptscriptstyle 1} }}\\
						\none[{\vphantom{A}_\emptyset}]},
												\superY{
						\yF{\sD{b_2}} \\
						\none[{\vphantom{A}_\emptyset}]\\
						\none[{\vphantom{A}_\emptyset}]\\
						\none[{\vphantom{A}_\emptyset}]\\
						b & \yF{\sT{b_{\scriptscriptstyle 1} }}},
												\superY{
						\none[{\vphantom{A}_\emptyset}]\\
						\yF{\sD{b_2}} \\
						b & \yF{\sT{b_{\scriptscriptstyle 1} }}\\
						\none[{\vphantom{A}_\emptyset}]\\
						\none[{\vphantom{A}_\emptyset}]},
												\superY{
						\none[{\vphantom{A}_\emptyset}]\\
						\yF{\sD{b_2}} \\
						\none[{\vphantom{A}_\emptyset}]\\
						b & \yF{\sT{b_{\scriptscriptstyle 1} }}\\
						\none[{\vphantom{A}_\emptyset}]}, \right.\nonumber \\
												&\phantom{\{}\left. 
						\superY{
						\none[{\vphantom{A}_\emptyset}]\\
						\yF{\sD{b_2}} \\
						\none[{\vphantom{A}_\emptyset}]\\
						\none[{\vphantom{A}_\emptyset}]\\
						b & \yF{\sT{b_{\scriptscriptstyle 1} }}},
												\superY{
						\none[{\vphantom{A}_\emptyset}]\\
						\none[{\vphantom{A}_\emptyset}]\\
						\yF{\sD{b_2}} \\
						b & \yF{\sT{b_{\scriptscriptstyle 1} }}\\
						\none[{\vphantom{A}_\emptyset}]},
												\superY{
						\none[{\vphantom{A}_\emptyset}]\\
						\none[{\vphantom{A}_\emptyset}]\\
						\yF{\sD{b_2}} \\
						\none[{\vphantom{A}_\emptyset}]\\
						b & \yF{\sT{b_{\scriptscriptstyle 1} }}},
												\superY{
						\none[{\vphantom{A}_\emptyset}]\\
						\none[{\vphantom{A}_\emptyset}]\\
						\none[{\vphantom{A}_\emptyset}]\\
						\yF{\sD{b_2}} \\
						b & \yF{\sT{b_{\scriptscriptstyle 1} }}}.
												\right\}
				\end{align}
				The elements of that set are then added to the $a$-filled diagram of $\ti\La$. The rows of the resulting diagrams are rearranged in proper order such that they can potentially be associated to a superpartition. 
				 				 Among the resulting diagrams, those that violate the conditions on superpartitions are rejected.
	In the remaining set of rearranged diagrams, we only keep the distinct ones. 
	The resulting set is
	\begin{align}
																					\mathcal{P} =	\left\{ 
					\superYbig{a & b & \yF{\sT{b_1}}&\yF{\sD{a_2}} \\
					a & \yF{\sT{a_1}} \\
					\yF{\sT{b_2}}&\yF{\sD{a_3}} },
										\superYbig{a & b & \yF{\sT{b_1}}&\yF{\sD{a_2}} \\
					a &  \yF{\sT{a_1}} \\
					\yF{\sT{b_2}} \\
					\yF{\sD{a_3}} },
										\superYbig{b & \yF{\sT{b_1}}&\yF{\sD{a_3}} \\
					a & \yF{\sT{a_1}} \\
					a & \yF{\sD{a_2}} \\
					\yF{\sT{b_2}} },
																\superYbig{b & \yF{\sT{b_1}}&\yF{\sD{a_3}} \\
					a& \yF{\sT{b_2}}&\yF{\sD{a_2}} \\
					a & \yF{\sT{a_1}} }.
					\right\}
				\end{align}
				\item[6.] Multiplicity of the monomial\\
				The first three diagrams of $\mathcal{P}$ do not have any part that is repeated (i.e., parts with the same number of boxes and same fermionic content). We thus have
				\beq N_{\Lambda \Omega}^{(\sTD{2}, \sT{1}, \sTD{0})}= N_{\Lambda \Omega}^{(\sTD{2}, \sT{1},\sT{0}, \sD{0})}= N_{\Lambda \Omega}^{(\sTD{1}, \sT{1}, \sD{1},\sT{0})}=1.\eeq 
In the fourth diagram, the first and the second parts have the same number of boxes and same fermionic content, so that \beq N_{\Lambda \Omega}^{(\sTD{1},\sTD{1},\sT{1})}=2.\eeq 
				\item[7.] Sign of the permutation\\
				Reading the content of the circles from left to right and top to bottom, we note the sequence of indexed symbols. We enumerate the sequences in the same order as they appear in $\mathcal{P}$: 
				\begin{align*}
					(b_1, a_2, a_1, b_2, a_3) &\xrightarrow{\text{odd permutation}} (a_1, a_2, a_3, b_1, b_2) \\
					(b_1, a_3, b_2, a_2, a_1) &\xrightarrow{\text{odd permutation}} (a_1, a_2, a_3, b_1, b_2) \\ 				(b_1, a_3, a_1, a_2, b_2) &\xrightarrow{\text{{odd} permutation}} (a_1, a_2, a_3, b_1, b_2)\\ (b_1, a_3, b_2, a_2, a_1) &\xrightarrow{\text{even permutation}} (a_1, a_2, a_3, b_1, b_2). 
				\end{align*}
			\end{enumerate}
			We finally have that 
			\begin{align}
				m_{\, \superYsmall{\,& \yT\\
				\,& \yC\\
				\yC}}
				m_{\, \superYsmall{\,& \yT\\
				\yT}} &=
				-m_{\, \superYsmall{\,& \,& \yTC\\
				\,& \yT\\
				\yTC}}
				-m_{\, \superYsmall{\,& \,& \yTC\\
				\,& \yT\\
				\yT\\
				\yC}}
				-m_{\, \superYsmall{\,& \yTC\\
				\,& \yT\\
				\,& \yC\\
				\yT}}
								+2m_{\, \superYsmall{\,& \yTC\\
				\,& \yTC\\
				\,& \yT}}.
			\end{align}
		 	\section{ {Completeness of {the multiplicative bases}}}
{Given that   the monomials form a basis, the standard strategy (cf. \cite{Macdonald1998}) for demonstrating that the elementary symmetric functions also form a basis goes through the demonstration that $e_{\la'}$ is lower triangular in its monomial expansion (i.e., it is expressed  in terms of the $m_\mu$ for $\mu\leq \la$). It then readily follows, from the involution $\om$, that  the $h_\la$ also form a basis. Finally, Newton's formulas, which allows to express the $h$'s
in terms of $p$'s, ensure that the power sums also form a basis.
\\

				But this route is blocked from the beginning  in our case, simply because there is no triangularity between the 
monomials and the $e_\Lambda$'s in the $\mathcal{N}=2$ case. We thus have to follow a different path. We will first prove that the power sums $p_\La$ form a basis and then use the standard argument (albeit  in inverse order) to deduce that the $h_\La$ and $e_\La$ do also form a basis.  
				For the first step, we proceed by induction: we suppose that the monomial decomposition of $p_\La$ is upper triangular.  The inductive step relies on the multiplicative character of the power sums: the power sum with the superpartition $\La^+$ which is $\La$ augmented by an extra part $\La_{L+1}$ (with $L=\ell(\La$)), that is $\La^+=\La\cup\La_{L+1}$, is obtained by multiplying $p_\La$ by $p_{\La_{L+1}}.$  Then, we use the fact that $p_{\La_{L+1}}=m_{\La_{L+1}}$ to transform the product of $p_{\La_{L+1}}$ times the monomial upper-triangular sum by means of the product rule for monomials described in Appendix B.
				The limitation of this procedure is that we cannot demonstrate that the elementary and homogeneous functions form $\mathbb {Z}$ bases but only the weaker statement that they are $\mathbb Q$ bases.}  But this is a rather minor detail in the present context.

	\subsection{Weight ordering on superpartitions} \label{subC1}
	We now introduce  a partial order on superpartitions which will 
allow us to prove that the transition matrix between the monomial and power sum bases is triangular. The order relies on the assignment of a weight to the various types of parts in a superpartition according to
\begin{align}\label{Wpar}
			w(a) &= a, \qquad
			w( \sT{a}) = 			w( \sD{a}) = a + \tfrac{1}{2} ,\qquad
			w( \sTD{a})= a + 1.
		\end{align}
Note that this is compatible with the ordering $\succ$ introduced in \eqref{orderP}: the parts in a superpartition are ordered by non-increasing value of their weight, i.e., 
		\beq \label{worder}\La=(\La_1,\La_2,\cdots ,\La_\ell): \qquad w(\La_i)\geq w(\La_{i+1}),\eeq
		 with the sole difference that in the $\succ$ ordering, we prioritize vertical circles.

		\begin{definition}[Weight Ordering] \label{def.weightOrder}
		 In terms of the weight assignment \eqref{Wpar}, we define the following ordering on superpartitions: 
		\begin{align}
			\Lambda > \Omega	\quad	&\text{iff} \quad
\quad \sum_{i=1}^n w(\Lambda_i) \geq \sum_{i=1}^n w(\Omega_i) \; \forall n
\quad \text{ and } \quad \exists  n \text{~such that~}
			 \sum_{i=1}^n w(\Lambda_i) > \sum_{i=1}^n w(\Omega_i) .
		\end{align}
		\end{definition}
		\begin{example}
		Let us illustrate the ordering with the following example:
		\begin{align}
			\superY{	\,&\,&\yTC}
			>
			\superY{	\,&\,&\yT\\
						\yC}
			>
			\superY{	\,&\yC\\
						\,&\yT}
			>
			\superY{	\,&\yT\\
						\,\\
						\yT}
			>			
			\superY{	\,\\
						\,\\
						\yTC}
			>
			\superY{	\,\\
						\,\\
						\yC\\
						\yT}
.		\end{align}
Note that two superpartitions $\Lambda$ and $\Omega$ can be distinct even though	\begin{align}
			\sum_{i=1}^n w(\Lambda_i) = \sum_{i=1}^n w(\Omega_i) \quad \forall n,
		\end{align}
in which case they are not comparable. For instance, 
		\begin{align}
			\superY{	\,&\,&\yT\\
						\yC}
			& \quad \text{and} \quad
			\superY{	\,&\,&\yC\\
						\yT}		
                 \end{align}
is such a pair of superpartitions.
		\end{example}
Observe that the weight ordering, 
which was only chosen for its simplicity,
does not reduce to the usual 
dominance ordering  for $\mathcal {N}={1}$ superpartitions when $\o m$ or $\u m$ vanishes.  For instance, the superpartition $(\bar 4,\bar 2,2,2,1)$ is
larger than the superpartition $(3,3,3,\bar 1,1,\bar 0)$ in the 
weight ordering while 
the two superpartitions are not comparable in the dominance ordering.

The total weight of a superpartition is defined in the obvious way as
		\begin{align}
			w(\Lambda) = \sum_{i=1}^{\ell (\Lambda)} w(\Lambda_i).
		\end{align}
Diagrammatically, the above weight assignments amounts to weighting boxes and circles as follows
		\begin{align}
			w \left(\makebox[10pt]{$\;\superY{\,}$} \right) = 1, \qquad			w\left(\makebox[10pt]{$\;\superY{\yT}$} \right) =			w\left(\makebox[10pt]{$\;\superY{\yC}$} \right) = \tfrac{1}{2}, \qquad
			w\left(\makebox[10pt]{$\;\superY{\yTC}$} \right) = 1 .
		\end{align}
In terms of the data $(n,\o  m,\u m)$, the total weight of a superpartition reads
		\begin{align}
			w(\Lambda) = n+\tfrac{1}{2}\sT{m} + \tfrac{1}{2}\sD{m}\quad \forall \Lambda \in \SPar{n|\sT{m},\sD{m}}.
		\end{align}
		Finally, note that if {the diagram of} $\Lambda$ can be obtained {from that of $\Gamma$} by lowering any symbol, then $\Gamma > \Lambda.$

		\begin{example}
			This last point is illustrated by the following series of diagrams where every diagram is obtained by the previous one by lowering one symbol (we consider \makebox[10pt]{$\;\superY{\yTC}$} as two symbols):
			\begin{align}
				\superY{
					\,&\,&\,&\yTC\\
					\yC
					}
				>
				\superY{
					\,&\,&\yTC\\
					\,\\
					\yC
					}
				>
				\superY{
					\,&\,&\yT\\
					\,&\yC\\
					\yC
					}
				>
				\superY{
					\,&\,\\
					\,&\yTC\\
					\yT
				}.
			\end{align}
		\end{example}
		 		
\subsection{Multiplicative bases}
				
	\begin{proposition}	\label{prop.preserveWeightDominance}
		Let $\Gamma$ and $\Lambda \in \SPar{n,\sT{m},\sD{m}}$ {be such that $\Gamma > \Lambda$. If the same part $\Lambda_{L+1}$, such that $\Lambda_{L+1}\leq \Lambda_{L}$, is added to both superpartitions 
then $\Gamma^+>\Lambda^+$, where}		\begin{align}\label{plus}
			\Lambda^+ &:= \La\cup(\La_{L+1})=(\Lambda_1, \Lambda_2, \cdots, \Lambda_{L+1})\\
			\Gamma^+ &:= \Gamma\cup(\La_{L+1})=(\Gamma_1, \cdots, \Gamma_{{k-1}}, \Lambda_{L+1}, \Gamma_{{k}}, \cdots, \Gamma_{\ell(\Gamma)}),
		\end{align}
with the further specification that, if such $\Gamma_k$ exists,
		\begin{align}
			w(\Lambda_{L+1})> w(\Gamma_{{k}}), \label{eq.LaBiggerGammak}
		\end{align}
meaning that if the part $\La_{L+1}$ is repeated in $\Gamma^+$, the added part is placed in the rightmost position.
	\end{proposition}
	\begin{proof} 	Suppose first that $k-1=\ell(\Gamma)$, that is, that $\Lambda_{L+1}$ is also the last
entry of $\Gamma^+$. Using $\Gamma > \Lambda$, which implies $\ell(\Gamma) \leq \ell(\Lambda)=L$,
we have 
\begin{align}
\sum_{i=1}^{n}w(\Gamma^+_i)=\sum_{i=1}^{n}w(\Gamma_i) \geq 
\sum_{i=1}^{n}w(\Lambda_i)=\sum_{i=1}^{n}w(\Lambda_i^+) \quad \forall n \leq \ell(\Gamma),
\end{align}
with at least one partial sum being strictly larger.
This immediately implies $\Gamma^+ > \Lambda^+$ since the only remaining partial sum is also such that 
\begin{align}
\sum_{i=1}^{\ell(\Gamma)+1}w(\Gamma_i^+)= w(\Gamma^+) = w(\Lambda^+) \geq \sum_{i=1}^{\ell(\Gamma)+1}w(\Lambda_i^+) \, . 
\end{align}

Now, suppose that $k-1 < \ell(\Gamma)$. Since $\Gamma>\Lambda$, we have
\begin{align}
\sum_{i=1}^{n}w(\Gamma^+_i)=\sum_{i=1}^{n}w(\Gamma_i) \geq 
\sum_{i=1}^{n}w(\Lambda_i)=\sum_{i=1}^{n}w(\Lambda_i^+) \quad \forall n \leq k-1
\end{align}
given that $k \leq \ell(\Gamma) \leq \ell(\Lambda)=L$.
Adding $\sum_{i=1}^{k-1}w(\Gamma_i)$ on each side of \eqref{eq.LaBiggerGammak} also results in
		\begin{align}
			\sum_{i=1}^{k-1}w(\Gamma_i) + w(\Lambda_{L+1}) &>\sum_{i=1}^{k}w(\Gamma_i)\quad
			\implies\quad \sum_{i=1}^{k}w(\Gamma^+_i) >\sum_{i=1}^{k}w(\Gamma_i)\geq  \sum_{i=1}^{k}w(\Lambda_i)= \sum_{i=1}^{k}w(\Lambda^+_i).
		\end{align}
	The ordering on the entries of superpartitions  implies that $\Gamma_i^+\geq \Gamma_i\, \forall\, i > k$, so we obtain
		\begin{align}
			\sum_{i=1}^{{n}}w(\Gamma^+_i) > \sum_{i=1}^{{n}}w(\Gamma_i)\geq \sum_{i=1}^n w(\Lambda_i)=w(\Lambda_i^+) \quad \text{whenever~} k \leq n \leq\ell(\Gamma). \label{eq.stronginq}
		\end{align}
where in the last step we used $\Gamma > \Lambda$ and $\ell(\Gamma)\leq L$.
As before, we also have
\begin{align}
\sum_{i=1}^{\ell(\Gamma)+1}w(\Gamma_i^+)= w(\Gamma^+) = w(\Lambda^+) \geq \sum_{i=1}^{\ell(\Gamma)+1}w(\Lambda_i^+) \, . 
\end{align}
which implies
	\begin{align}
			\sum_{i=1}^n w(\Gamma_i^+) \geq \sum_{i=1}^n w(\Lambda_i^+) \quad \forall n .\label{ineq.gamaplusBigLambdaplus}
		\end{align}
Since, as we have seen, the inequality is strict when $n=k$, we have that $\Gamma^+ > \Lambda^+$.
	\end{proof}

{
\begin{example}
	Let us illustrate the previous proposition with an example. Let $\Gamma = (\sT{4}, 2)$ and $\Lambda= (\sT{3}, 3)$. We see that $\Gamma > \Lambda$. Now let $\Lambda^+= (\sT{3},3 ,\sTD{2})$. So we have 
	\begin{align}
		\Gamma^+ = (\sT{4}, 2) \cup (\sTD{2}) = \superY{\,&\,&\,&\,&\yT\\ \,&\,&\yTC\\ \,&\,}, \quad \Lambda^+ = (\sT{3}, 3) \cup (\sTD{2}) = \superY{\,&\,&\,&\yT\\ \,&\,&\, \\ \,&\,&\yTC}.
	\end{align}
	We see that $\Gamma^+$ is still larger than $\Lambda^+$ in the weight ordering. 
\end{example}
}

	\begin{proposition}\label{prop:powersumistriangular}
		The expansion of the power sums basis in the monomial basis is upper triangular.  To be more precise, 
		\begin{align}\label{uptri}
			&p_{\Lambda} = a_\Lambda m_\Lambda + \sum_{\Gamma > \Lambda} a_{\Lambda, \Gamma}m_\Gamma,
		\qquad \text{with}\qquad
			a_\Lambda= {\prod_{\substack{i=1\\ \Lambda_{i}\neq \Lambda_{i+1} }}^L}			 n_\Lambda(\Lambda_i)!.
		\end{align}
		where 
	$n_\Lambda(\Lambda_i)$ is the number of occurrences of $\Lambda_i$ in the superpartition $\Lambda$.
	\end{proposition}
	\begin{proof}
		{Let us define the superpartitions $\La,\La^+,\Gamma$ and $\Gamma^+$ as in \eqref{plus}.}
We proceed by induction.	Suppose the correctness of \eqref{uptri} and consider	\begin{align}
			p_{\Lambda^+}=p_{\Lambda}p_{\Lambda_{L+1}} = a_{\Lambda}m_\Lambda m_{\Lambda_{L+1}} + \sum_{\Gamma > \Lambda} a_{\Lambda,\Gamma}m_{\Gamma}m_{\Lambda_{L+1}}. \label{eq.plambdaplus}
		\end{align}
		Let us first focus on the $a_{\Lambda}m_\Lambda m_{\Lambda_{L+1}}$ term. Given the rules of the product of monomials presented in Appendix \ref{sect:monomialProd}, we know that this product of monomials will result in  {an expression of the form (where $\mathcal P$ is defined by the algorithm described there)				} 
				\begin{align}
			m_\Lambda m_{\Lambda_{L+1}} = N_{\Lambda\Lambda_{L+1}}^{\Lambda^+}m_{\Lambda^+} + \sum_{\Om \in \mathcal{P} | \Om \neq \Lambda^+} \epsilon_{\Lambda \Lambda_{L+1}}^\Om N_{\Lambda \Lambda_{L+1}}^\Om m_\Om.
		\end{align}
		Now the second term is a sum over every possible ways of adding {(with the  $+$ operation described in \eqref{addition})} the part $\Lambda_{L+1}$ to every other part of $\Lambda$.  This means that we can obtain $\Lambda^+$ from {such a  superpartition, say  $\Om$, by lowering a number of symbols in the diagrammatic representation of the latter}. This amounts to say (see last comment of Subsection~\ref{subC1}) that all these superpartitions $\Omega$
are higher in the weight ordering than $\La^+$. This means that 
		\begin{align}
			m_\Lambda m_{\Lambda_{L+1}} = N_{\Lambda\Lambda_{L+1}}^{\Lambda^+}m_{\Lambda^+} + \text{ higher terms}.
		\end{align}
		Moreover,   from the monomial product rules, we know that 
		\begin{align}
			N_{\Lambda\Lambda_{L+1}}^{\Lambda^+} = n_{\Lambda^+}(\Lambda_{L+1}).
		\end{align}
		Now, let us focus on the sum in \eqref{eq.plambdaplus}. The product of monomials that is summed upon is
		\begin{align}
			m_{\Gamma}m_{\Lambda_{L+1}} \quad \text{with } \Gamma > \Lambda.
		\end{align}
		From the previous argument, we know that the result of the product will contain the monomial of superpartition $\Gamma$ to which we {adjoin (with the $\cup$ operation)} $\Lambda_{L+1}$ as a new part (giving $\Gamma^+$ defined in \eqref{plus}) followed by a sum of terms higher in the weight ordering. This means that
we have		\begin{align}
			m_{\Gamma}m_{\Lambda_{L+1}} \propto m_{\Gamma^+} + \text{ higher terms}.
		\end{align}
		From proposition \ref{prop.preserveWeightDominance} we know that 
		if $ \Gamma > \Lambda$, then $\Gamma^{+} > \Lambda^{+}.$
Equation \eqref{eq.plambdaplus} can thus be written as
		\begin{align}
			p_{\Lambda^+}&= n_{\Lambda^+}(\Lambda_{L+1}){\prod_{\substack{i=1\\ \Lambda_{i}\neq \Lambda_{i+1} }}^L}  			n_\Lambda(\Lambda_i)!\,m_{\Lambda^{+}} + \text{ higher terms}\\
							&= {\prod_{\substack{i=1\\ \Lambda_{i}\neq \Lambda_{i+1} }}^{L+1} }			n_{\Lambda^+}(\Lambda^+_i)!\,m_{\Lambda^{+}} + \text{ higher terms},
		\end{align}
{which completes the proof of the inductive step. The proof of \eqref{uptri} is completed by noticing that the result is obviously true for a  single-part superpartition: $	p_{(\Lambda_1)} = m_{(\Lambda_1)}$.
}
	\end{proof}
\begin{example}
		Let us illustrate the proof with an example. Let $\Lambda=(\sT{2},\sD{1})$ and $\Lambda^+=(\sT{2},\sD{1},1)$. We have
		\begin{align}
			p_{
				\superYsmall{
				\,&\,&\yT\\
				\,&\yC
				}
			}
			=
			m_{
				\superYsmall{
				\,&\,&\yT\\
				\,&\yC}
			}
			+m_{
				\superYsmall{
				\,&\,&\,&\yTC
				}
			}.
		\end{align} 
		Then, we have
		\begin{align}
			p_{
				\superYsmall{
				\,&\,&\yT\\
				\,&\yC \\
				\,
				}
			}=
			p_{
				\superYsmall{
				\,&\,&\yT\\
				\,&\yC \\
				}
			}p_{\superYsmall{\,}}
			&=
			m_{
				\superYsmall{
				\,&\,&\yT\\
				\,&\yC}
			}
			m_{\superYsmall{\,}}
			+m_{
				\superYsmall{
				\,&\,&\,&\yTC
				}
				}	
			m_{\superYsmall{\,}}\\
			&=
			(	
				m_{
				\superYsmall{
				\,&\,&\yT\\
				\,&\yC\\
				\,}
				}+
				m_{
				\superYsmall{
				\,&\,&\yT\\
				\,&\,&\yC}
				}+
				m_{
				\superYsmall{
				\,&\,&\,&\yT\\
				\,&\yC}
				}
			)	
			+(
			m_{
				\superYsmall{
				\,&\,&\,&\yTC\\
				\,}
				}	+
			m_{
				\superYsmall{
				\,&\,&\,&\,&\yTC}
				}			
			).
		\end{align}
	\end{example} 
	\begin{proposition}
The power sums $p_\La$, for $\La\in\text{SPar}$, form a basis of
$\mathscr{P}^{S_\infty}(\mathbb{Q})$.		
	\end{proposition}
	\begin{proof}
	From proposition \ref{prop:powersumistriangular}, we know that 
	{the transition matrix $M(p,m)_{\Lambda,\Gamma}$ (which describes the monomial expansion of the power sums)} is upper triangular, that is
	\begin{align}
		M(p,m)_{\Lambda,\La} \neq 0 \quad \forall \Lambda \qquad\text{and}\qquad
		M(p,m)_{\Lambda,\Gamma} = 0 \quad \forall \,\Gamma \leq \Lambda.
	\end{align}
The triangular nature of the transition matrix and the absence of zero entries on its diagonal ensure its invertibility over $\mathbb Q$, i.e., that $M(m,p)$ is well defined. This implies that any monomial symmetric function can be expressed as a linear combination over $\mathbb Q$ of the power sums. Since the monomials form a basis of the ring of symmetric superfunctions over $\mathbb Q$ (given that they do over $\mathbb Z$), so do the power sums.
	\end{proof}

	\begin{proposition} \label{prop:homoisabasis}
The homogeneous symmetric functions $h_\La$, for $\La\in\text{SPar}$, form a basis of
$\mathscr{P}^{S_\infty}(\mathbb{Q})$.		
	\end{proposition}
	\begin{proof}
The relation $h_n$ vs $p_n$ in Table~\ref{tab2} ensures 
that the $h$'s can be written in terms of the $p$'s and vice versa.
Since, as we have seen, the $p$'s form a multiplicative basis of $\mathscr{P}^{S_\infty}(\mathbb{Q})$, any
element of $\mathscr{P}^{S_\infty}(\mathbb{Q})$ can be written in terms of
$h$'s.  Given that the $h$'s are indexed by superpartitions (as are the $p$'s)
and that the
homogeneous degree of $h_\Lambda$ is the same as that
of $p_\Lambda$, the $h$'s also form a basis
of $\mathscr{P}^{S_\infty}(\mathbb{Q})$.
\end{proof}
	\begin{proposition}
	The elementary symmetric functions $e_\La$, for $\La\in\text{SPar}$, form a basis of
$\mathscr{P}^{S_\infty}(\mathbb{Q})$.
	\end{proposition}
	\begin{proof}
		The proposition is an immediate consequence of the homomorphism $\hat{\omega}$ defined in \eqref{eq:involution} and Proposition \ref{prop:homoisabasis}.
	\end{proof}



\begin{thebibliography}{10}
\bibitem{Note}
		It should be pointed out that symmetric functions in several sets of commuting variables invariant under a diagonal action of the symmetric group have been considered long ago by MacMahon (see for instance \cite{MacMahon}) whose results are briefly reviewed in \cite{Gessel} and exemplified in the case of two sets of variables. In particular, possible generalizations of the classical bases (the monomial, elementary, power-sum and complete bases) are presented.   Even in the case of
two sets of variables, these generalizations do not correspond to those obtained by removing the 
anticommuting variables in the ${\mathcal N}=1$ classical bases (see for instance \cite{BLM}).  This is in part because the symmetrization process underlying our 
construction is more intricate than the MacMahon one due to 
the nilpotent character of the anticommuting variables.
One might also wonder whether the anticommutativity of the $\ta$ and $\phi$ variables (i.e., the minus sign in \eqref{tap}) does play a key role in the present context. It turns out not to be crucial at the level of the classical bases considered here. However, it is dictated by physical considerations and will be relevant for the applications mentioned in the concluding remarks.
\end{thebibliography}
\end{document}